\documentclass[11pt]{article}

\usepackage{amsmath, amsthm, amssymb}
\usepackage{amsfonts}
\usepackage{mathrsfs}
\usepackage{enumitem}
\usepackage{bbm}
\usepackage{pifont,xspace,fullpage,epsfig, wrapfig}
\newcommand{\myparagraph}[1]{\smallskip\noindent {\bf {#1}}}

\DeclareMathSymbol{\N}{\mathbin}{AMSb}{"4E}
\DeclareMathSymbol{\Z}{\mathbin}{AMSb}{"5A}
\DeclareMathSymbol{\R}{\mathbin}{AMSb}{"52}
\DeclareMathSymbol{\Q}{\mathbin}{AMSb}{"51}
\def\E{\operatorname*{\mathbb{E}}}
\def\F{\operatorname*{\mathbb{F}}}

\newcommand{\remove}[1]{}
\newcommand{\sumnl}{\sum\nolimits}

\newcommand{\BBB}{\mathcal B}
\newcommand{\BbB}{\mathscr{B}}
\newcommand{\CCC}{\mathcal C}
\newcommand{\DDD}{\mathcal D}
\newcommand{\FFF}{\mathcal F}
\newcommand{\HHH}{\mathcal H}
\newcommand{\HhH}{\mathscr{H}}

\newcommand{\PPP}{\mathcal P}

\newcommand{\UUU}{\mathcal U}

\newcommand{\qual}{\mathop{\rm q}}
\newcommand{\error}{{\rm error}}
\newcommand{\db}{S}
\newcommand{\dbGood}{S_{\rm good}}
\newcommand{\size}{\operatorname{\rm size}}
\newcommand{\VC}{\operatorname{\rm VC}}
\newcommand{\maj}{\operatorname{\rm maj}}
\newcommand{\MAJ}{\operatorname{\rm MAJ}}
\newcommand{\RepDim}{\operatorname{\rm RepDim}}
\newcommand{\DRepDim}{\operatorname{\rm DRepDim}}
\newcommand{\opt}{\operatorname{\rm OPT}}

\newcommand{\point}{\operatorname*{\tt POINT}}
\newcommand{\set}[1]{\left\{ #1 \right\}}

\newtheorem{theorem}{Theorem}[section]
\newtheorem{example}[theorem]{Example}
\newtheorem{lemma}[theorem]{Lemma}
\newtheorem{definition}[theorem]{Definition}
\newtheorem{remark}[theorem]{Remark}
\newtheorem{proposition}[theorem]{Proposition}
\newtheorem{claim}[theorem]{Claim}

\newtheorem{corollary}[theorem]{Corollary}

\newtheorem{observation}[theorem]{Observation}

\newcommand{\lemref}[1]{Lemma~\ref{lem:#1}}

\newcommand{\defref}[1]{Definition~\ref{def:#1}}
\newcommand{\exmref}[1]{Example~\ref{exm:#1}}

\begin{document}
\begin{titlepage}
\title{Characterizing the Sample Complexity of Private Learners\thanks{A preliminary version of this paper appeared in~\cite{BNS13}. Research partially supported by the Israel Science Foundation (grants No.\ 938/09 and 2761/12) and by the Frankel Center for Computer Science.}}
\author{Amos Beimel \quad Kobbi Nissim \quad Uri Stemmer \\
Dept. of Computer Science \\
Ben-Gurion University of the Negev\\
{\tt \{beimel|kobbi|stemmer\}@cs.bgu.ac.il}}

\date{\today}
\maketitle
\setcounter{page}{0} \thispagestyle{empty}

\begin{abstract}

In 2008, Kasiviswanathan et al.\ defined {\em private learning} as a combination of PAC learning and differential privacy~\cite{KLNRS08}. Informally, a private learner is applied to a collection of labeled individual information and outputs a hypothesis while preserving the privacy of each individual. Kasiviswanathan et al.\ gave a generic construction of private learners for (finite) concept classes, with sample complexity logarithmic in the size of the concept class. This sample complexity is higher than what is needed for non-private learners, hence leaving open the possibility that the sample complexity of private learning may be sometimes significantly higher than that of non-private learning.

We give a combinatorial characterization of the sample size sufficient and necessary to privately learn a class of concepts. This characterization is analogous to the well known characterization of the sample complexity of non-private learning in terms of the VC dimension of the concept class.  We introduce the notion of {\em probabilistic representation} of a concept class, and our new complexity measure $\RepDim$ corresponds to the size of the smallest probabilistic representation of the concept class.

We show that any private learning algorithm for a concept class $\CCC$ with sample complexity $m$ implies $\RepDim(\CCC)=O(m)$, and that there exists a private learning algorithm with sample complexity $m=O(\RepDim(\CCC))$. We further demonstrate that a similar characterization holds for the database size needed for privately computing a large class of optimization problems and also for the well studied problem of private data release.

\end{abstract}

\end{titlepage}

\tableofcontents
\setcounter{page}{0} \thispagestyle{empty}
\newpage

\section{Introduction}

Motivated by the observation that learning generalizes many of the analyses applied to large collections of data, Kasiviswanathan el al.~\cite{KLNRS08} defined in 2008 {\em private learning} as a combination of probably approximately correct (PAC) learning~\cite{Valiant84} and differential privacy~\cite{DMNS06}. A PAC learner is given a collection of labeled examples (sampled according to an unknown probability distribution and labeled according to an unknown concept) and generalizes the labeled examples into a hypothesis $h$ that should predict with high accuracy the labeling of fresh examples taken from the same unknown distribution and labeled with the same unknown concept.

The privacy requirement is that the choice of $h$ preserves differential privacy of sample points. Intuitively this means that this choice should not be significantly affected by any particular sample. Differential privacy is increasingly accepted as a standard for rigorous privacy and recent research has shown that differentially private variants exists to many analyses. We refer the reader to surveys of Dwork~\cite{Dwork09,Dwork11}.

The sample complexity required for learning a concept class $\CCC$ determines the amount of {\em labeled} data needed for learning a concept $c\in\CCC$. It is well known that the sample complexity of learning a concept class $\CCC$ (non-privately) is proportional to a complexity measure of the class $\CCC$ knowns as the VC-dimension~\cite{VC,BEHW,EHKV}. Kasiviswanathan et al.~\cite{KLNRS08} proved that a private learner exists for every {\em finite} concept class. The proof is via a generic construction that exhibits sample complexity logarithmic in the size of the concept class. The VC-dimension of a concept class is bounded by this quantity (and significantly lower for some interesting concept classes), and hence the results of~\cite{KLNRS08} left open the possibility that the sample complexity of private learning may be significantly higher than that of non-private learning.

In analogy to the characterization of the sample complexity of (non-private) PAC learners via the VC-dimension, we give a combinatorial characterization of the sample size sufficient and necessary for private PAC learners. Towards obtaining this characterization, we introduce the notion of {\em probabilistic representation} of a concept class. We note that our characterization, as the VC-dimension characterization, ignores the computation required by the learner. Some of our algorithms are, however, computationally efficient.

\subsection{Related Work}
\label{sec:related}

We start with a short description of prior work on the sample complexity of private learning. To simplify the exposition, we ignore dependencies on the error, confidence and privacy parameters by considering them constants for this and the following section. The dependency on these parameters would be made explicit in the later sections of the paper.

Recall that the sample complexity of non-private learners for a class of functions $\CCC$ is proportional to the VC-dimension of the class~\cite{BEHW,EHKV} -- a combinatorial measure of the class that is equal to the size of the largest set of inputs that is shattered by the class. This characterization, as ours, ignores the computation required by the learner.

Kasiviswanathan et al.~\cite{KLNRS08} showed, informally, that every finite concept 
class $\CCC$ can be learned privately (ignoring computational complexity). Their construction is based on the exponential mechanism of McSherry and Talwar~\cite{MT07}, and the $O(\ln |\CCC|)$ bound on sample complexity  results from the union bound argument used in the analysis of the exponential mechanism. Computationally efficient learners were shown to exist by Blum et al.~\cite{BDMN05} for all concept classes that can be efficiently learned in the {\em statistical queries} model. Kasiviswanathan et al.~\cite{KLNRS08} showed an example of a concept class -- the class of parity functions -- that is not learnable in the statistical queries model but can be learned privately and efficiently. These positive results suggest that many ``natural'' computational learning tasks that are efficiently learned non-privately can be learned privately and efficiently. 

Beimel et al.~\cite{BKN10} studied the sample complexity of private learning. They examined the concept class of point functions $\point_d$ where each concept evaluates to one on exactly one point of the domain and to zero otherwise. Note that the VC-dimension of $\point_d$ is one. Beimel et al.\ proved 
lower bounds on the sample complexity of {\em properly} and privately learning the class $\point_d$ (and related classes), implying that the VC dimension of a class does not characterize the sample complexity of private proper learning. On the other hand, they observed that the sample complexity can be improved for {\em improper} private learners whenever there exists a smaller hypothesis class $\HHH$ that represents $\CCC$ in the sense that for every concept $c \in \CCC$ and for every distribution on the examples, there is a hypothesis $h \in \HHH$ that is close to $c$. Using the exponential mechanism to choose among the hypotheses in $\HHH$ instead of $\CCC$, the sample complexity is reduced to $\ln |\HHH|$ (this is why the {\em size} of the representation $\HHH$ is defined to be $\ln |\HHH|$). For some classes this can dramatically improve the sample complexity, e.g., for the class $\point_d$ (defined in \exmref{point}), the sample complexity is improved from $O(\ln |\point_d|)=O(d)$ to $O(\ln d)$. Using other techniques, Beimel et al.\ showed that the sample complexity of learning $\point_d$ can be reduced even further to $O(1)$, hence showing the largest possible gap between proper and non proper private learning. Such a gap
does not exists for non-private learning.

Chaudhuri and Hsu~\cite{CH11} studied the sample complexity needed for private
learning infinite concept classes when the data is drawn from a continuous
distribution. They showed that under these settings there exists a simple
concept class for which any proper learner that uses a finite number of
examples and guarantees differential privacy fails to satisfy accuracy
guarantee for at least one data distribution. This implies that
the results of Kasiviswanathan et al.~\cite{KLNRS08} do not extend to
infinite hypothesis classes. Interestingly, our results imply an improper
private algorithm for an infinite extension of the class $\point$ (that is,
a class over the natural numbers of all boolean functions that return 1 on
exactly one number).

Chaudhuri and Hsu~\cite{CH11} also study learning algorithms that are only
required to protect the privacy of the labels (and do not necessarily protect
the privacy of the examples themselves). They prove upper bounds and lower
bounds on the sample complexity of such algorithms. In particular, they
prove a lower bound on the sample complexity using the doubling dimension
of the disagreement metric of the hypothesis class with respect to the
unlabeled data distribution. This result does not imply our
characterization as the privacy requirement in protecting the labels is
much weaker than protecting the sample point and the label.

A line of research (started in~\cite{Schapire90}) that is very relevant to
our paper is boosting learning algorithms, that is, taking learning
algorithms that have a big classification error and producing a learning
algorithm with small error.  Dwork et al.~\cite{DRV10} show how to
privately boost accuracy, that is, given a {\em private} learning algorithms that
have a big classification error, they produce a {\em private} learning
algorithm with small error. In \lemref{noParams}, we show how to boost the
accuracy $\alpha$ for probabilistic representations. This gives an
alternative private boosting, whose proof is simpler.  However, as it uses 
the exponential mechanism, it is (generally) not computationally efficient.

%%Other tools for private learning were studied in a few
%%%papers; such tools include, for example, private logistic
%%regression~\cite{CM08} and private empirical risk minimization~\cite{DMS11}.

\subsection{Our Results}

Beimel et al.~\cite{BKN10} showed how to use a representation of a class to privately learn it.
We make an additional step in improving the sample complexity
by considering a {\em probabilistic} representation of a concept class $\CCC$. Instead of one
collection $\HHH$ representing $\CCC$, we consider a list of collections
$\HHH_1,\dots,\HHH_r$ such that for every $c\in \CCC$ and every
distribution on the examples, if we sample a collection $\HHH_i$ from the
list, then with high probability there is a hypothesis $h \in \HHH_i$ that
is close to $c$. To privately learn $\CCC$, the learning algorithm first samples $i\in\{1,\ldots,r\}$ and then uses the exponential mechanism to select a hypothesis from $\HHH_i$. 
This reduces the sample complexity to $O(\max_{i} \ln |\HHH_i|)$;
the {\em size} of the probabilistic representation is hence defined to be $\max_{i} \ln |\HHH_i|$.

We show that for $\point_d$ there exists a probabilistic representation of size $O(1)$. This results in a private learning algorithm  with sample complexity $O(1)$, matching a different private algorithm for $\point_d$ presented in~\cite{BKN10}. Our new algorithm offers
some improvement in the sample complexity compared to the algorithm of~\cite{BKN10} when considering the learning and privacy parameters.
Furthermore, our algorithm can be made computationally efficient without making any computational hardness assumptions, while the efficient version in~\cite{BKN10} assumes the existence of one-way functions. Finally, it is conceptually simpler and in particular it avoids the sub-sampling technique used in~\cite{BKN10}.

One can ask if there are private learning algorithms with smaller sample 
complexity than the size of the smallest probabilistic representation. We
show that the answer is no --- the size of the smallest probabilistic representation is a lower bound on the
sample complexity. Thus, the size of the smallest probabilistic
representation of a class $\CCC$, which we call the {\em representation dimension} and denote by $\RepDim(\CCC)$,
characterizes (up to constants) the sample size necessary and sufficient for privately learning the class $\CCC$. We also show that for concepts defined over a finite domain, the difference between the sizes of the best deterministic and probabilistic representation is bounded. Namely, 
that if $\CCC$ is a concept class over the domain $\set{0,1}^d$, then there exists a deterministic representation of $\CCC$ of size $O(\RepDim(\CCC)+\ln d)$. Thus, for classes whose smallest deterministic representation is of size  $\omega(\ln d)$, the size of the smallest deterministic representation characterizes the sample complexity of private learning of the class.

%7th paragraph: some consequences of our results

The notion of probabilistic representation applies not only to private
learning, but also to optimization problems. We consider a scenario where
there is a domain $X$, a database $\db$ of $m$ records, each taken from the
domain $X$, a set of solutions $\FFF$, and a quality function $\qual:X^*\times
\FFF \rightarrow [0,1]$ that we wish to maximize. If the exponential mechanism
is used for (approximately) solving the problem, then the size of the
database should be $\Omega(\ln |\FFF|)$ in order to achieve a reasonable
approximation.  Using our notions of a representation of $\FFF$ and of a
probabilistic representation of $\FFF$, one can reduce the size of the minimal
database without paying too much in the quality of the
solution. Interestingly, a similar notion to representation, called
``solution list algorithms'', was considered in~\cite{BCNW08} for
constructing secure protocols for search problems while leaking only a few
bits on the input. Curiously, their notion of leakage is
very different from that of differential privacy.

We give two examples of such optimization problems. First, an example
inspired by~\cite{BCNW08}: each record in the database is a clause with
exactly 3 literals and we want to find an assignment satisfying at least
7/8 fraction of the clauses while protecting the privacy of the clauses.  A
construction of~\cite{BCNW08} yields a deterministic representation for
this problem where the size of the database can be much smaller. Using a
probabilistic representation, we can give a good assignment even for
databases of constant size. This example is a simple instance of a
scenario, where each individual has a preference on the solution and we
want to choose a solution maximizing the number of individuals whose
preference are met, while protecting the privacy of the preference.  Another
example of optimization is sanitization, where given a database we want to
publish a synthetic database, which gives a similar utility as the original
database while protecting the privacy of the individual records of the
database. Using our techniques, we study the minimal database size
for which sanitization gives reasonable performance with respect to a given
family of queries.

\myparagraph{Open Problem.} We still do not know the relation between this
dimension and the VC dimension. By Sauer's Lemma, if $\CCC$ is a concept
class over $\set{0,1}^d$, then the number of functions in $\CCC$ is at most
$\exp(d\cdot\VC(\CCC))$. By~\cite{KLNRS08}, there is a private learning
algorithm for $\CCC$ whose sample size is $O(d \cdot\VC(\CCC))$, thus, the
probabilistic representation dimension of $\CCC$ is $O(d\cdot
\VC(\CCC))$. We do not know if there is a class $\CCC$ such that
$\RepDim(\CCC) \gg \VC(\CCC)$. A candidate for such separation appears
in~\cite{BBKN12}.

\section{Preliminaries}\label{sec:defs}

\myparagraph{Notation.} We use $O_{\gamma}(g(n))$ as a shorthand for $O(h(\gamma) \cdot g(n))$ for some
non-negative function $h$. Given a set $\BBB$ of cardinality $r$, and a distribution $\PPP$ on $\{1,2,\ldots,r\}$, we use the notation $b\in_{\PPP}\BBB$ to denote a random element of $\BBB$ chosen according to $\PPP$.

\subsection{Preliminaries from Privacy} 
A database is a vector $\db = (z_1,\dots,z_m)$ over a domain $X$, where each entry $z_i \in \db$ represents information contributed by one individual. Databases $\db_1$ and $\db_2$ are called {\em neighboring} if they differ in exactly one entry. An algorithm preserves differential privacy if neighboring databases induce nearby outcome distributions. Formally,

\begin{definition}[Differential Privacy~\cite{DMNS06}] \label{def:eps-dp} A randomized algorithm $A$ is $\epsilon$-differentially private if for all neighboring databases $\db_1,\db_2$, and for all sets $\mathcal{F}$ of outputs,
\begin{eqnarray}
\label{eqn:diffPrivDef}
  & \Pr[A(\db_1) \in \mathcal{F}] \leq \exp(\epsilon) \cdot \Pr[A(\db_2) \in \mathcal{F}].  &
\end{eqnarray}
The probability is taken over the random coins of $A$. 
\end{definition}
An immediate consequence of the definition is that for {\em any} two databases $\db_1,\db_2 \in X^m$, and for all sets $\mathcal{F}$ of outputs, 
$$\Pr[A(\db_1) \in \mathcal{F}] \geq \exp(-\epsilon m) \cdot \Pr[A(\db_2) \in \mathcal{F}].$$

\subsection{Preliminaries from Learning Theory}  
Let $X_d=\{0,1\}^d$. A concept $c:X_d\rightarrow \{0,1\}$ is a function that labels {\em examples} taken from the domain $X_d$ by either 0 or 1.  A \emph{concept class} $\CCC$ over $X_d$ is a class of concepts mapping $X_d$ to $\{0,1\}$. 

PAC learning algorithms are given examples sampled according to an unknown
probability distribution $\DDD$ over $X_d$, and labeled according to an
unknown {\em target} concept $c\in\CCC$. The {\em generalization error} of a
hypothesis $h:X_d\rightarrow\{0,1\}$ is defined as 
$$\error_{\DDD}(c,h)=\Pr_{x \in_{\DDD} X_d}[h(x)\neq c(x)].$$
For a labeled sample $\db=(x_i,y_i)_{i=1}^m$, the {\em empirical error} of $h$ is
$$\error_S(h) = \frac{1}{m} |\{i : h(x_i) \neq y_i\}|.$$

\begin{definition}
An {\em $\alpha$-good} hypothesis for $c$ and $\DDD$ is a hypothesis $h$ such that $\error_{\DDD}(c,h)\leq\alpha$. 
\end{definition}

\begin{definition}[PAC Learning~\cite{Valiant84}] \label{def:PAC}
Algorithm $A$ is an {\em $(\alpha,\beta)$-PAC learner} for a concept
class $\CCC$ over $X_d$ using hypothesis class $\HHH$ and sample size $m$ if for all 
concepts $c \in \CCC$, all distributions $\DDD$ on $X_d$,
given an input of $m$ samples $\db =(z_1,\ldots,z_m)$, where $z_i=(x_i,c(x_i))$ and $x_i$
are drawn i.i.d.\ from $\DDD$, algorithm $A$ outputs a
hypothesis $h\in \HHH$ satisfying
$$\Pr[\error_{\DDD}(c,h)  \leq \alpha] \geq 1-\beta.$$
The probability is taken over the random choice of
the examples in $\db$ according to $\DDD$ and the coin tosses of the learner $A$.
\end{definition}

\begin{definition}
An algorithm satisfying Definition~\ref{def:PAC} with $\HHH\subseteq\CCC$ is called a {\em proper} PAC learner; otherwise it is called an {\em improper} PAC learner.
\end{definition}

\subsection{Private Learning}
As a private learner is a PAC learner, its outcome hypothesis should also be a good predictor of labels. Hence, the privacy requirement from a private learner is not that an application of the hypothesis $h$ on a new sample (pertaining to an individual) should leak no information about the sample.
\begin{definition}[Private PAC Learning~\cite{KLNRS08}]
\label{def:private-general}
Let $A$ be an algorithm that gets an input $\db =(z_1,\ldots,z_m)$. Algorithm $A$ is an {\em $(\alpha,\beta,\epsilon)$-PPAC learner} for a concept
class $\CCC$ over $X_d$ using hypothesis class $\HHH$ and sample size $m$ if
\begin{description}
\item{\sc Privacy.} Algorithm $A$ is $\epsilon$-differentially private
  (as formulated in \defref{eps-dp});
\item{\sc Utility.} Algorithm $A$ is an {\em $(\alpha,\beta)$-PAC learner} for $\CCC$ using $\HHH$ and sample size $m$ (as formulated in \defref{PAC}).
\end{description}
\end{definition}

\subsection{The Exponential Mechanism}\label{expMech}
We next describe the exponential mechanism of McSherry and Talwar~\cite{MT07}. We present its private learning variant; however, it can be used in more general scenarios. The goal here is to chooses a hypothesis $h\in \HHH$ approximately minimizing the empirical error. The choice is probabilistic, where the probability mass that is assigned to each hypothesis decreases exponentially with its empirical error.

\noindent
\begin{center}\fbox{
\parbox{37pc}{
Inputs: a privacy parameter $\epsilon$, a hypothesis class $\HHH$, and $m$ labeled samples $S=(x_i,y_i)_{i=1}^m$.
\begin{enumerate}[itemindent=3pt, topsep=5pt, itemsep=1pt,align=left,leftmargin=*,
        labelsep=3pt]
  \item $\forall h\in\HHH$ define $q(S,h)=|\{i:h(x_i)=y_i\}|$.
	\item Randomly choose $h \in \HHH$ with probability
	$$\frac{\exp\left(\epsilon \cdot q(S,h) /2 \right)}{\sum_{f\in\HHH}\exp\left(\epsilon \cdot q(S,f) /2 \right)}.$$
\end{enumerate}
}}\end{center}
	
\begin{proposition}
Denote $\hat e\triangleq\min_{f\in \HHH}\{\error_S(f)\}$. The probability that the exponential mechanism outputs a hypothesis $h$ such that $\error_S(h)>\hat e + \Delta$ is at most $|\HHH| \cdot \exp(-\epsilon \Delta m /2)$.
Moreover, The exponential mechanism is $\epsilon$ differentially private.
\end{proposition}

\subsection{Concentration Bounds}\label{bounds}
Let $X_1,\dots,X_n$ be independent random variables where $\Pr[X_i=1]=p$ and $\Pr[X_i=0]=1-p$ for some $0<p<1$. Clearly, $\E[\sumnl_i{X_i}]=pn$. Chernoff and Hoeffding bounds show that the sum is concentrated around this expected value:
\begin{align}
&\Pr\left[\sumnl_i{X_i}>(1+\delta)pn\right]\leq \exp\left(-pn\delta^2/3\right) \;\;\text{ for } \delta>0,\nonumber\\
&\Pr\left[\sumnl_i{X_i}<(1-\delta)pn\right]\leq \exp\left(-pn\delta^2/2\right) \;\;\text{ for } 0<\delta<1,\nonumber\\
&\Pr\left[\left|\sumnl_i{X_i}-pn\right|>\delta\right]\leq 2\exp\left(-2\delta^2/n\right) \;\,\;\;\text{ for } \delta\geq0.\nonumber
\end{align}
The first two inequalities are known as the multiplicative Chernoff bounds~\cite{chern}, and the last inequality is known as the Hoeffding bound~\cite{hoeff}.

\section{The Sample Complexity of Private Learners} \label{sec:learn}

In this section we present a combinatorial measure of a concept class $\CCC$ that characterizes the sample complexity necessary and sufficient for privately learning $\CCC$. The measure is a {\em probabilistic representation} of the class $\CCC$. We start with the notation of deterministic representation from \cite{BKN10}.

\begin{definition}[\cite{BKN10}]\label{def:rep}
A hypothesis class $\HHH$ is an $\alpha$-representation for a class $\CCC$ if for every $c \in \CCC$ and every distribution $\DDD$ on $X_d$ there exists a hypothesis $h \in \HHH$ such that $\error_\DDD(c,h)\leq \alpha$. 
\end{definition}

\begin{example}[$\point_d$]
\label{exm:point}
For $j\in X_d$, define $c_j:X_d \rightarrow \{0,1\}$ as $c_j(x)=1$ if $x=j$, and $c_j(x)=0$ otherwise. Define $\point_d=\{c_j\}_{j\in X_d}$. In~\cite{BKN10} it was shown that for $\alpha<1/2$, every $\alpha$-representation for $\point_d$ must be of cardinality at least $d$, and that an $\alpha$-representation $\HHH_d$ for $\point_d$ exists where $|\HHH_d|=O(d/\alpha^2)$.
\end{example}

The above representation can be used for non-private learning, by taking a big enough sample and finding a hypothesis $h\in\HHH_d$ minimizing the empirical error. For {\em private} learning it was shown in \cite{BKN10} that a sample of size $O_{\alpha,\beta,\epsilon}(\log|\HHH_d|)$ suffices, with a  learner that employs the exponential mechanism to choose a hypothesis from $\HHH_d$.

\begin{definition}\label{def:drepdim}
For a hypothesis class $\HHH$ we denote  $\size(\HHH)=\ln|\HHH|$. We define the Deterministic Representation Dimension of a concept class $\CCC$ as
$$\DRepDim(\CCC) = \min\Big\{   \size(\HHH)   : \HHH \text{ is a } \frac{1}{4} \text{-representation for } \CCC \Big\}.$$
\end{definition}

\begin{remark}
Choosing $\frac{1}{4}$ is arbitrary; we could have chosen any (smaller than $\frac{1}{2}$) constant.
\end{remark}

\begin{example}
By the results of \cite{BKN10}, stated in the previous example, $\DRepDim(\point_d)=\theta(\ln(d))$.
\end{example}

We are now ready to present the notion of a probabilistic representation. The idea behind this notion is that we have a list of hypothesis classes, such that for every concept $c$ and distribution $\DDD$, if we sample a hypothesis class from the list, then with high probability it contains a hypothesis that is close to $c$.

\begin{definition}\label{def:prep}
Let $\PPP$ be a distribution over $\{1,2, \ldots ,r\}$, and let $\HhH=\{\HHH_1,\HHH_2, \ldots ,\HHH_r\}$ be a family of hypothesis classes (every $\HHH_i \in \HhH$ is a set of boolean functions). We say that $(\HhH,\PPP)$ is an $(\alpha,\beta)$-probabilistic representation for a class $\CCC$ if for every $c \in \CCC$ and every distribution $\DDD$ on $X_d$:
$$\Pr_{\PPP}\left[\exists h \in \HHH_i \;\; s.t. \;\; \error_\DDD(c,h)\leq \alpha\right]\geq 1-\beta.$$
The probability is over randomly choosing a set $\HHH_i\in_{\PPP}\HhH$.
\end{definition}

\begin{remark}
As we will see in Section\ref{sec:equivalence}, the existence of such a probabilistic representation $(\HhH,\PPP)$ for a concept class $\CCC$ implies the existence of a private learning algorithm for $\CCC$ with sample complexity that depends on the cardinality of the hypothesis classes $\HHH_i\in\HhH$. The sample complexity will not depend on $r=|\HhH|$. Nevertheless, in Section~\ref{sec:relationships} we will see that there always exists a probabilistic representation in which $r$ is bounded.
\end{remark}

\begin{example}[$\point_d$]\label{example:pointRep}
In Section \ref{sec:proofs} we construct for every $\alpha$ and every $\beta$ a pair $(\HhH,\PPP)$ that $(\alpha,\beta)$-probabilistically represents the class $\point_d$, where $\HhH$ contains all the sets of at most $\frac{4}{\alpha}\ln(1/\beta)$ boolean functions.
\end{example}

\begin{definition}\label{def:repdim}
Let $\HhH=\{\HHH_1,\HHH_2, \ldots ,\HHH_r\}$ be a family of hypothesis classes. We denote $|\HhH|=r$, and $\size(\HhH)=\max\{  \;  \ln|\HHH_i| : \HHH_i \in \HhH  \;  \}$. We define the Representation Dimension of a concept class $\CCC$ as
$$\RepDim(\CCC) = \min\left\{  \;  \size(\HhH)  \;  : \;
\begin{array}{l}
\exists \PPP \text{ s.t. } (\HhH,\PPP) \text{ is a }\\
(\frac{1}{4},\frac{1}{4})\text{-probabilistic}\\
 \text{representation for } \CCC
\end{array}\right\}.$$
\end{definition}

\begin{remark}
Choosing $\alpha=\beta=\frac{1}{4}$ is arbitrary; we could have chosen any two (smaller than $\frac{1}{2}$) constants.
\end{remark}

\begin{example}[$\point_d$]\label{example:pointDim}
The $\size$ of the probabilistic representation mentioned in Example \ref{example:pointRep} is $\ln(\frac{4}{\alpha}\ln(1/\beta))$.
Placing $\alpha=\beta=\frac{1}{4}$, we see that the Representation Dimension of $\point_d$ is constant.
\end{example}

\subsection{Equivalence of $(\alpha,\beta)$-Probabilistic Representation and Private Learning}\label{sec:equivalence}
We now show that $\RepDim(\CCC)$ characterizes the sample complexity of private learners. We start by showing in Lemma \ref{lem:equivalence2} that an $(\alpha,\beta)$-probabilistic representation of $\CCC$ implies a private learning algorithm whose sample complexity is the size of the representation. We then show in Lemma \ref{lem:equivalence1} that if there is a private learning algorithm with sample complexity $m$, then there is probabilistic representation of $\CCC$ of size $O(m)$; this lemma implies that $\RepDim(\CCC)$ is a lower bound on the sample complexity. Recall that $\RepDim(\CCC)$ is the size of the smallest probabilistic representation for $\alpha=\beta=1/4$. Thus, to complete the proof we show in Lemma \ref{lem:noParams} that a probabilistic representation with $\alpha=\beta=1/4$ implies a probabilistic representation for arbitrary $\alpha$ and $\beta$.

\begin{lemma}\label{lem:equivalence2}
If there a exists pair $(\HhH,\PPP)$ that  $(\alpha,\beta)$-probabilistically represents a class $\CCC$, then for every $\epsilon$ there exists an algorithm $A$ that $(6\alpha,4\beta,\epsilon)$-PPAC learns $\CCC$ with a sample size
$m=O\left(  \frac{1}{\alpha\epsilon} ( \size(\HhH) + \ln(\frac{1}{\beta}) )  \right)$.
\end{lemma}

\begin{proof}
Let $(\HhH,\PPP)$ be an $(\alpha,\beta)$-probabilistic representation for the class $\CCC$, and consider the following algorithm $A$:
$$\boxed{
\begin{array}{l}
\text{Inputs: } S=(x_i,y_i)_{i=1}^m \text{, and a privacy parameter } \epsilon.\\
{\begin{array}{ll}
1. & \text{Randomly choose } \HHH_i \in_\PPP \HhH .\\
2. & \text{Choose } h \in \HHH_i \text{ using the exp. mechanism with } \epsilon.\\
\end{array}}\\
\end{array}}$$
By the properties of the exponential mechanism, $A$ is $\epsilon$-differentially private. We will show that with sample size $m=O\left(  \frac{1}{\alpha\epsilon} ( \size(\HhH) + \ln(\frac{1}{\beta}) )   \right)$, algorithm $A$ is a $(6\alpha,4\beta)$-PAC learner for $\CCC$. Fix some $c\in \CCC$ and $\DDD$, and define the following 3 good events:
\begin{enumerate}[label=$E_{\arabic*}$]
\item $\HHH_i$ chosen in step 1 contains at least one hypothesis $h$ s.t. $\error_S(h)\leq2\alpha$.
\item For every $h\in\HHH_i$ s.t. $\error_S(h)\leq3\alpha$, it holds that $\error_{\DDD}(c,h)\leq6\alpha$
\item The exponential mechanism chooses an $h$ such that $\error_S(h) \leq \alpha + \min_{f\in \HHH_i}\left\{\error_S(f)\right\}$.
\end{enumerate}

We first show that if those 3 good events happen, algorithm $A$ returns a $6\alpha$-good hypothesis.
Event $E_1$ ensures the existence of a hypothesis $f\in\HHH_i$ s.t. $\error_S(f)\leq2\alpha$. Thus, event $E_1 \cap E_3$ ensures algorithm $A$ chooses (using the exponential mechanism) a hypothesis $h\in\HHH_i$ s.t. $\error_S(h)\leq3\alpha$. Event $E_2$ ensures therefore that this $h$ obeys $\error_{\DDD}(c,h)\leq6\alpha$.

We will now show that those 3 events happen with high probability. As $(\HhH,\PPP)$ is an $(\alpha,\beta)$-probabilistic representation for the class $\CCC$, the chosen $\HHH_i$ contains a hypothesis $h$ s.t.  $\error_{\DDD}(c,h)\leq\alpha$ with probability at least $1-\beta$; by the Chernoff bound with probability at least $1-\exp(-m\alpha/3)$ this hypothesis has empirical error at most $2\alpha$. Event $E_1$ happens with probability at least $(1-\beta)(1-\exp(-m\alpha/3))>1-(\beta+\exp(-m\alpha/3))$, which is at least $(1-2\beta)$ for $m\geq\frac{3}{\alpha}\ln(1/\beta)$.

Using the Chernoff bound, the probability that a hypothesis $h$ s.t. $\error_{\DDD}(c,h)>6\alpha$ has empirical error $\leq3\alpha$ is less than $\exp(-m\alpha3/4)$. Using the union bound, the probability that there is such a hypothesis in $\HHH_i$ is at most $|\HHH_i|\cdot\exp(-m\alpha3/4)$. Therefore, $\Pr[ E_2 ]\geq1-|\HHH_i|\cdot\exp(-m\alpha3/4)$. For $m\geq\frac{4}{3\alpha}( \ln(\frac{|\HHH_i|}{\beta}) )$, this probability is at least $(1-\beta)$.

The exponential mechanism ensures that the probability of event $E_3$ is at least $1-|\HHH_i| \cdot \exp(-\epsilon \alpha m /2)$ (see Section \ref{expMech}), which is at least $(1-\beta)$ for $m \geq \frac{2}{\alpha \epsilon} \ln(\frac{|\HHH_i|}{\beta})$.

All in all, by setting $m=\frac{3}{\alpha\epsilon} ( \size(\HhH) + \ln(\frac{1}{\beta}) ) $ we ensure that the probability of $A$ failing to output a $6\alpha$-good hypothesis is at most $4\beta$.
\end{proof}

We will demonstrate the above lemma with two examples:

\begin{example}[Efficient learner for $\point_d$]\label{example:pointAlgo}
As described in Example \ref{example:pointRep}, there exists an $(\HhH,\PPP)$ that  $(\alpha/6,\beta/4)$-probabilistically represents the class $\point_d$, where $\size(\HhH)=O_{\alpha,\beta,\epsilon}(1)$.
By Lemma \ref{lem:equivalence2}, there exists an algorithm that $(\alpha,\beta,\epsilon)$-PPAC learns $\CCC$ with sample size $m=O_{\alpha,\beta,\epsilon}(1)$.

The existence of an algorithm with sample complexity $O(1)$ was already proven in \cite{BKN10}. Moreover, assuming the existence of oneway functions, their learner is efficient. Our constructions yields an efficient learner, without assumptions. To see this, consider again algorithm $A$ presented in the above proof, and note that as $\size(\HhH)$ is constant, step 2 could be done in constant time. 
Step 1 can be done efficiently as we can efficiently sample a set $\HHH_i\in_{\PPP}\HhH$.
In Claim \ref{claim:pointEff} we initially construct a probabilistic representation in which the description of every hypothesis is exponential in $d$.
The representation is than revised using pairwise independence to yield a representation in which every hypothesis $h$ has a short description, and given $x$ the value $h(x)$ can be computed efficiently.
\end{example}

\begin{example}[$\point_{\N}$]\label{example:pointRatio}
Consider the class $\point_{\N}$,  which is exactly like $\point_d$, only over the natural numbers.
By results of \cite{CH11,BKN10}, it is impossible to properly PPAC learn the class $\point_{\N}$.
Our construction can yield an (inefficient) improper private learner for $\point_{\N}$ with $O_{\alpha,\beta,\epsilon}(1)$ samples. The details are deferred to Section \ref{sec:proofs}.
\end{example}

The next lemma shows that a private learning algorithm implies a probabilistic representation. This lemma can be used to lower bound the sample complexity of private learners.

\begin{lemma}\label{lem:oldEquivalence1}
If there exists an algorithm $A$ that $(\alpha,\frac{1}{2},\epsilon)$-PPAC learns a concept class $\CCC$ with a sample size $m$, then there exists a pair $(\HhH,\PPP)$ that $(\alpha,1/4)$-probabilistically represents the class $\CCC$ such that $\size(\HhH) = O\left(m\epsilon\right)$.
\end{lemma}

\begin{proof}
Let $A$ be an $(\alpha,\frac{1}{2},\epsilon)$-PPAC learner for a class $\CCC$ using hypothesis class $\FFF$ whose sample size is $m$. For a target concept $c \in \CCC$ and a distribution $\DDD$ on $X_d$, we define $G$ as the set of all hypotheses $h\in\FFF$ such that $\error_{\DDD}(c,h)\leq\alpha$.
Fix some $c\in\CCC$ and a distribution $\DDD$ on $X_d$.
As $A$ is an $(\alpha,\frac{1}{2})$-PAC learner,
$\Pr_{\DDD , A}\left[ A(\db) \in G \right]\geq \frac{1}{2}$,
where the probability is over $A$'s randomness and over sampling the examples in $\db$ (according to $\DDD$).
Therefore, there exists a database $\db$ of $m$ samples such that
$\Pr_A\left[ A(\db) \in G \right]\geq \frac{1}{2}$,
where the probability is only over the randomness of $A$.
As $A$ is $\epsilon$-differentially private,
$ \Pr_A\left[ A(\vec{0}) \in G \right] \geq e^{-m\epsilon} \cdot \Pr_A\left[ A(\db) \in G \right]\geq \frac{1}{2} e^{-m\epsilon} $,
where $\vec{0}$ is a database with $m$ zeros.\footnote{Choosing $\vec{0}$ is arbitrary; we could have chosen any database.} That is, $\Pr_A\left[ A(\vec{0}) \notin G \right] \leq 1-\frac{1}{2} e^{-m\epsilon}$. Now, consider a set $\HHH$ containing the outcomes of $2 \ln(4) e^{m\epsilon} $ executions of $A(\vec0)$. The probability that $\HHH$ does not contain an $\alpha$-good hypothesis is at most $(1-\frac{1}{2} e^{-m\epsilon})^{2 \ln(4) e^{m\epsilon} } \leq \frac{1}{4}$. Thus,
$\HhH = \left\{ \HHH \subseteq \FFF \; : \; |\HHH| \leq  2 \ln(4) e^{m\epsilon}  \right\}$,
and $\PPP$, the distribution induced by $A(\vec0)$, are an $(\alpha,1/4)$-probabilistic representation for class $\CCC$. It follows that $\size(\HhH) = \max\{  \;  \ln|\HHH| : \HHH \in \HhH  \;  \} = \ln(2\ln(4))+ m\epsilon$.
\end{proof}

The above lemma yields a lower bound of $\Omega\left(\frac{1}{\epsilon}\RepDim(\CCC)\right)$ on the sample complexity of private learners for a concept class $\CCC$.
To see this, fix $\alpha\leq\frac{1}{4}$ and let $A$ be an $(\alpha,\frac{1}{2},\epsilon)$-PPAC learner for $\CCC$ with sample size $m$. By the above lemma, there exists a pair $(\HhH,\PPP)$ that $(\alpha,1/4)$-probabilistically represents $\CCC$ s.t. $\size(\HhH) = \ln(2\ln(4))+ m\epsilon$. Therefore, by definition, $\RepDim(\CCC)\leq\ln(2\ln(4))+ m\epsilon$. Thus, $m\geq\frac{1}{\epsilon}(\RepDim(\CCC)-\ln(2\ln(4)))=\Omega\left(\frac{1}{\epsilon}\RepDim(\CCC)\right)$.

In order to refine this lower bound (and incorporate $\alpha$ in it), we will need a somewhat stronger version of this lemma:

\begin{lemma}\label{lem:equivalence1}
Let $\alpha\leq1/4$. If there exists an algorithm $A$ that $(\alpha,\frac{1}{2},\epsilon)$-PPAC learns a concept class $\CCC$ with a sample size $m$, then there exists a pair $(\HhH,\PPP)$ that $(1/4,1/4)$-probabilistically represents the class $\CCC$ such that $\size(\HhH) = O\left(m\epsilon\alpha\right)$.
\end{lemma}

\begin{proof}
Let $A$ be an $(\alpha,\frac{1}{2},\epsilon)$-PPAC learner for the class $\CCC$ using hypothesis class $\FFF$ whose sample size is $m$.
Without loss of generality, we can assume that $m\geq\frac{3\ln(4)}{4\alpha}$ (since A can ignore part of the sample).
For a target concept $c \in \CCC$ and a distribution $\DDD$ on $X_d$, we define
$$G_{\DDD}^\alpha = \{  h\in\FFF : \error_{\DDD}(c,h)\leq\alpha \}.$$
Fix some $c\in\CCC$ and a distribution $\DDD$ on $X_d$, and define the following distribution $\widetilde{\DDD}$ on $X_d$:
$$ \Pr_{\widetilde{\DDD}}[x]=\begin{cases}
    1-4\alpha+4\alpha\cdot \Pr_{\DDD}[x], & x=0^d.\\
    4\alpha\cdot \Pr_{\DDD}[x], & x \neq 0^d.
  \end{cases}  $$
Note that for every $x\in X_d$,
\begin{eqnarray}\label{eq:tildeD}
\Pr_{\widetilde{\DDD}}[x]\geq4\alpha\cdot \Pr_{\DDD}[x].
\end{eqnarray}
As $A$ is an $(\alpha,\frac{1}{2})$-PAC learner, it holds that
$$\Pr_{\widetilde{\DDD} , A}\left[ A(\db) \in G_{\widetilde{\DDD}}^\alpha \right]\geq \frac{1}{2},$$
where the probability is over $A$'s randomness and over sampling the examples in $\db$ (according to $\widetilde{\DDD}$). In addition, by inequality (\ref{eq:tildeD}), every hypothesis $h$ with $\error_{\DDD}(c,h)>1/4$ has error strictly greater than $\alpha$ under $\widetilde{\DDD}$:
\begin{eqnarray*}
\error_{\widetilde{\DDD}}(c,h) \geq 4\alpha\cdot\error_{\DDD}(c,h)>\alpha.
\end{eqnarray*}
So, every $\alpha$-good hypothesis for $c$ and $\widetilde{\DDD}$ is a $\frac{1}{4}$-good hypothesis for $c$ and $\DDD$. That is,
$ G_{\widetilde{\DDD}}^\alpha \subseteq G_{\DDD}^{1/4}$.
Therefore, $\Pr_{\widetilde{\DDD} , A}\left[ A(\db) \in G_{\DDD}^{1/4} \right]\geq \frac{1}{2}$.

We say that a database $\db$ of $m$ labeled examples is {\em good} if the unlabeled example $0^d$ appears in $\db$ at least $(1-8\alpha)m$ times.
Let $\db$ be a database constructed by taking $m$ i.i.d. samples from $\widetilde{\DDD}$, labeled by $c$.
By the Chernoff bound, $\db$ is good with probability at least $1-\exp(-4\alpha m /3)$. Hence,

$$\Pr_{\widetilde{\DDD} , A}\left[   ( A(\db) \in G_{\DDD}^{1/4} ) \wedge ( \db {\rm \;is\;good } )    \right]\geq \frac{1}{2}-\exp(-4\alpha m /3)\geq\frac{1}{4}.$$

Therefore, there exists a database $\dbGood$ of $m$ samples that contains the unlabeled sample $0^d$ at least $(1-8\alpha)m$ times, and
$\Pr_A\left[ A(\dbGood) \in G_{\DDD}^{1/4}  \right]\geq\frac{1}{4}$,
where the probability is only over the randomness of $A$. All of the examples in $\dbGood$ (including the example $0^d$) are labeled by $c$.

For $\sigma\in\{0,1\}$, denote by $\vec{0}_\sigma$ a database containing $m$ copies of the example $0^d$ labeled as $\sigma$. As $A$ is $\epsilon$-differentially private, and as the target concept $c$ labels the example $0^d$ by either $0$ or $1$, for at least one $\sigma\in\{0,1\}$ it holds that
\begin{align}\label{eqn:epsilonDelta}  
\Pr_A[A(\vec{0}_\sigma)\in G_{\DDD}^{1/4}]
&\geq \exp(-8\alpha\epsilon m) \cdot \Pr_A\left[ A(\dbGood) \in G_{\DDD}^{1/4}  \right] \nonumber\\
&\geq \exp(-8\alpha\epsilon m) \cdot 1/4.
\end{align}
That is, $\Pr_A[A(\vec{0}_{\sigma})\notin G_{\DDD}^{1/4}] \leq 1-\frac{1}{4} e^{-8\alpha\epsilon m}$. Now, consider a set $\HHH$ containing the outcomes of $4 \ln(4) e^{8\alpha\epsilon m} $ executions of $A(\vec{0}_0)$, and the outcomes of $4 \ln(4) e^{8\alpha\epsilon m} $ executions of $A(\vec{0}_1)$. The probability that $\HHH$ does not contain a $\frac{1}{4}$-good hypothesis for $c$ and $\DDD$ is at most $(1-\frac{1}{4} e^{-8\alpha\epsilon m})^{4 \ln(4) e^{8\alpha\epsilon m} } \leq \frac{1}{4}$. Thus,
$\HhH = \left\{ \HHH \subseteq \FFF \; : \; |\HHH| \leq  2 \cdot 4 \ln(4) e^{8\alpha\epsilon m}  \right\}$,
and $\PPP$, the distribution induced by $A(\vec{0}_0)$ and $A(\vec{0}_1)$, are a $(1/4,1/4)$-probabilistic representation for the class $\CCC$.
Note that the value $c(0^d)$ is unknown, and can be either 0 or 1. Therefore the construction uses the two possible values (one of them correct).

It holds that $\size(\HhH) = \max\{  \;  \ln|\HHH| : \HHH \in \HhH  \;  \} = \ln(8\ln(4))+ 8\alpha\epsilon m = O\left(m\epsilon\alpha\right) $.
\end{proof}

Lemma \ref{lem:noParams} shows how to construct a probabilistic representation for an arbitrary $\alpha$ and $\beta$ from a probabilistic representation with $\alpha=\beta=1/4$; in other words we boost $\alpha$ and $\beta$. The proof of this lemma is combinatorial. It allows us to start with a private learning algorithm with constant $\alpha$ and $\beta$, move to a representation, use the combinatorial boosting, and move back to a private algorithm with small $\alpha$ and $\beta$. This should be contrasted with the private boosting of~\cite{DRV10} which is algorithmic and more complicated (however, the algorithm of Dwork et al.~\cite{DRV10} is computationally efficient).

We first show how to construct a probabilistic representation for arbitrary $\beta$ from a probabilistic representation with $\beta=1/4$.

\begin{claim}\label{claim:boostB}
For every concept class $\CCC$ and for every $\beta$, there exists a pair $(\HhH,\PPP)$ that $(1/4,\beta)$-probabilistically represents $\CCC$ where $\size(\HhH) \leq \RepDim(\CCC) + \ln\ln(1/\beta)$.
\end{claim}

\begin{proof}
Let $\beta<1/4$, and let $(\HhH^0,\PPP^0)$ be a $(\frac{1}{4},\frac{1}{4})$- probabilistic representation for $\CCC$ with $\size(\HhH^0)=\RepDim(\CCC) \triangleq k_0$ (that is, for every $\HHH^0_i \in \HhH^0$ it holds that $|\HHH^0_i| \leq e^{k_0}$).
Denote $\HhH^0=\{ \HHH^0_1,\HHH^0_2,\ldots,\HHH^0_r  \}$, and consider the following family of hypothesis classes:
$$\HhH^1= \left\{ \HHH^0_{i_1} \cup \cdots \cup \HHH^0_{i_{\ln(1/\beta)}} \; : \; 1\leq i_1\leq \dots \leq i_{\ln(1/\beta)} \leq r \right\}.$$
Note that for every $\HHH^1_i\in\HhH^1$ it holds that $|\HHH^1_i| \leq \ln(1/\beta) e^{k_0}$ and so $\size(\HhH^1) \triangleq k_1 \leq k_0+\ln\ln(1/\beta)$.
We will now show an appropriate distribution $\PPP^1$ on $\HhH^1$ s.t. $(\HhH^1,\PPP^1)$ is a $(\frac{1}{4},\beta)$-probabilistic representation for $\CCC$. 
To this end, consider the following process for randomly choosing an $\HHH^1 \in \HhH^1$:

$$\boxed{
\begin{array}{ll}
1.&\text{Denote } M=\ln(1/\beta)\\
2.& \text{For } i=1, \ldots, M:\\
& \hspace{5 mm} \text{Randomly choose } \HHH^0_i \in_{\PPP_0} \HhH^0 .\\
3.&\text{Return } \HHH^1=\bigcup_{i=1}^M \HHH^0_i.
\end{array}}$$

The above process induces a distribution on $\HhH^1$, denoted as $\PPP^1$. As $\HhH^0$ is a 
$(\frac{1}{4},\frac{1}{4})$-probabilistic representation for $\CCC$, we have that
\begin{eqnarray*}
&&\Pr_{\PPP^1}\left[ \nexists h\in \HHH^1 \; s.t. \; \error_\DDD(c,h) \leq 1/4\right] = \\
&&=\prod_{i=1}^M{ \Pr_{\PPP^0}\left[ \nexists h\in \HHH^0_i \; s.t. \; \error_\DDD(c,h) \leq 1/4\right] } \leq \\
&&\leq \left(\frac{1}{4}\right)^M \leq \beta.
\end{eqnarray*}
\end{proof}

\begin{lemma}\label{lem:noParams}
For every concept class $\CCC$, every $\alpha$, and every $\beta$, there exists $(\HhH,\PPP)$ that $(\alpha,\beta)$-probabilistically represents $\CCC$ where
$$\size(\HhH) = O\Big(    \ln(\frac{1}{\alpha}) \cdot \big(  \RepDim(\CCC) + \ln\ln\ln(\frac{1}{\alpha}) + \ln\ln(\frac{1}{\beta})  \big)  \Big).$$
\end{lemma}

\begin{proof}
Let $\CCC$ be a concept class, and let $(\HhH^1,\PPP^1)$ be a $(\frac{1}{4},\beta/T)$-probabilistic representation for $\CCC$ (where $T$ will be set later). By Claim \ref{claim:boostB}, such a representation exists with $\size(\HhH^1)\triangleq k_1 \leq \RepDim(\CCC)+\ln\ln(T/\beta)$. We use $\HhH^1$ and $\PPP^1$ to create an $(\alpha,\beta)$- probabilistic representation for $\CCC$. We begin with two notations:\\
\begin{enumerate}[itemindent=0pt, topsep=0pt, itemsep=5pt,align=left,leftmargin=*,
        labelsep=5pt]
  \item For $T$ hypotheses $h_1,\ldots,h_T$ we denote by $\maj_{h_1,\ldots,h_T}$ the majority hypothesis. That is, $\maj_{h_1,\ldots,h_T}(x)=1$ if and only if $|\{ h_i \; : \; h_i(x)=1  \}| \geq T/2$.
	\item For $T$ hypothesis classes $\HHH_1,\ldots,\HHH_T$ we denote\\ $\MAJ(\HHH_1,\ldots,\HHH_T)=\Big\{  \maj_{h_1,\ldots,h_T} \; : \; \forall_{1\leq i \leq T} \; h_i\in \HHH_i  \Big\}$.\\
\end{enumerate}
Consider the following family of hypothesis classes:
$$\HhH=\bigg\{   \MAJ(\HHH_{i_1},\ldots,\HHH_{i_T}) \; : \; \HHH_{i_1},\ldots,\HHH_{i_T} \in \HhH^1  \bigg\}.$$
Moreover, denote the distribution on $\HhH$ induced by the following random process as $\PPP$:
$$\boxed{
\begin{array}{l}
\text{For } j=1, \ldots, T:\\
\hspace{5 mm} \text{Randomly choose } \HHH_{i_j}\in_{\PPP^1} \HhH^1\\
\text{Return } \MAJ(\HHH_{i_1},\ldots,\HHH_{i_T}).\\
\end{array}}$$
Next we show that $(\HhH,\PPP)$ is an $(\alpha,\beta)$-probabilistic representation for $\CCC$: For a fixed pair of a target concept $c$ and a distribution $\DDD$, randomly choose $\HHH_{i_1},\ldots,\HHH_{i_T} \in_{\PPP^1} \HhH^1$. We now show that with probability at least $(1-\beta)$ the set $\MAJ(\HHH_{i_1},\ldots,\HHH_{i_T})$ contains at least one $\alpha$-good hypothesis for $c,\DDD$.

To this end, denote $\DDD_1=\DDD$ and consider the following thought experiment, inspired by the Adaboost Algorithm of \cite{Adaboost}:
\begin{center}
\noindent\fbox{
\parbox{25pc}{
For $t=1.\ldots,T$:\\

\noindent\begin{enumerate}[itemindent=3pt, topsep=0pt, itemsep=0pt,align=left,leftmargin=*,
        labelsep=3pt]
  \item Fail if $\HHH_{i_t}$ does not contain a $\frac{1}{4}$-good hypothesis for $c,\DDD_t$.
	\item Denote by $h_t \in \HHH_{i_t}$ a $\frac{1}{4}$-good hypothesis for $c,\DDD_t$.
	\item $\DDD_{t+1}(x)=\begin{cases}
    2 \DDD_t(x),\;\; \text{if } h_t(x) \neq c(x).\\
    \left(1-\frac{\error_{\DDD_t}(c,h_t)}{1-\error_{\DDD_t}(c,h_t)} \right) \DDD_t(x), \;\; \text{otherwise}.
  \end{cases}$
\end{enumerate}}}\end{center}

\noindent Note that as $\DDD_1$ is a probability distribution on $X_d$; the same is true for $\DDD_2,\DDD_3,\ldots,\DDD_T$. As $(\HhH^1,\PPP^1)$ is a $(\frac{1}{4},\beta/T)$-probabilistic representation for $\CCC$, the failure probability of every iteration is at most $\beta/T$. Thus (using the union bound), with probability at least $(1-\beta)$ the whole thought experiment will succeed, and in this case we show that the error of $h_{\rm fin}=\maj_{h_1,\ldots,h_T}$ is at most $\alpha$.

Consider the set $R=\{x \; : \; h_{\rm fin}(x) \neq c(x) \} \subseteq X_d$. This is the set of points on which at least $T/2$ of $h_1,\ldots,h_T$ err. Next consider the partition of $R$ to the following sets:
$$R_t=\left\{x\in R \; : \; \big( h_t(x) \neq c(x) \big) \wedge \big( \forall_{i>t} \, h_i(x)=c(x) \big)  \right\}.$$
That is, $R_t$ contains the points $x\in R$ on which $h_t$ is last to err. Clearly $\DDD_t(R_t) \leq 1/4$, as $R_t$ is a subset of the set of points on which $h_t$ errs. Moreover,
\begin{eqnarray*}
\DDD_t(R_t) &\geq& \DDD_1(R_t) \cdot 2^{T/2} \cdot \left(1-\frac{\error_{\DDD_t}(c,h_t)}{1-\error_{\DDD_t}(c,h_t)}\right)^{t-T/2} \\
&\geq& \DDD_1(R_t) \cdot 2^{T/2} \cdot  \left(1-\frac{1/4}{1-1/4}\right)^{t-T/2} \\
&\geq& \DDD_1(R_t) \cdot 2^{T/2} \cdot  \left(1-\frac{1/4}{1-1/4}\right)^{T/2} \\
&=& \DDD(R_t) \cdot \left(\frac{4}{3}\right)^{T/2},
\end{eqnarray*}
so,
$$\DDD(R_t) \leq \DDD_t(R_t) \cdot \left(\frac{4}{3}\right)^{-T/2} \leq \frac{1}{4} \cdot \left(\frac{4}{3}\right)^{-T/2}.$$
Finally, \begin{eqnarray*}
&&\error_\DDD (c,h_{\rm fin})=\DDD(R)=\sum_{t=T/2}^T{\DDD(R_t)}\leq \\
&&\leq \frac{T}{2} \cdot \frac{1}{4} \cdot \left(\frac{4}{3}\right)^{-T/2}=\frac{T}{8} \cdot \left(\frac{4}{3}\right)^{-T/2}.
\end{eqnarray*}
Choosing $T=14 \ln(\frac{2}{\alpha})$, we get that $\error_\DDD (c,h_{\rm fin}) \leq \alpha$.
\remove{$$
\error_\DDD (c,h_{\rm fin}) \leq \frac{T}{8} \cdot \left(\frac{4}{3}\right)^{-T/2} = 
\frac{14}{8} \cdot \ln(\frac{2}{\alpha}) \cdot \left(\frac{4}{3}\right)^{-7 \ln(\frac{2}{\alpha})} = 
\frac{14}{8} \cdot \ln(\frac{2}{\alpha}) \cdot \left( \left( \left(\frac{4}{3}\right)^{3.5} \right)^{-\ln(\frac{2}{\alpha})} \right)^2 \leq$$
$$
\leq \frac{14}{8} \cdot \ln(\frac{2}{\alpha}) \cdot \left( e^{-\ln(\frac{2}{\alpha})} \right)^2 =
\frac{14}{8} \cdot \ln(\frac{2}{\alpha}) \cdot \left( \frac{\alpha}{2} \right)^2 \leq
\frac{14}{8} \cdot (\frac{2}{\alpha}) \cdot \left( \frac{\alpha}{2} \right)^2 = \frac{7}{8}\alpha \leq \alpha
$$}
Hence, $(\HhH,\PPP)$ is an $(\alpha,\beta)$-probabilistic representation for $\CCC$. Moreover, for every $\HHH_i \in \HhH$ we have that $|\HHH_i| \leq \left(e^{k_1}\right)^T$, and so
\begin{eqnarray*}
&&\size(\HhH) \leq k_1 \cdot T \leq \big( \RepDim(\CCC)+\ln\ln(T/\beta) \big)T \\
&&= O\Big(    \ln(\frac{1}{\alpha}) \cdot \big(  \RepDim(\CCC) + \ln\ln\ln(\frac{1}{\alpha}) + \ln\ln(\frac{1}{\beta})  \big)  \Big).
\end{eqnarray*}
\end{proof}
The next theorem states the main result of this section -- $\RepDim$ characterizes the sample complexity of private learning.

\begin{theorem}\label{thm:RepDim}
Let $\CCC$ be a concept class. $\widetilde{\Theta}_{\beta}\left(\frac{\RepDim(\CCC)}{\alpha\epsilon}\right)$ samples are necessary and sufficient for the private learning of the class $\CCC$.
\end{theorem}

\begin{proof}
Fix some $\alpha\leq 1/4,\beta\leq 1/2$, and $\epsilon$. By Lemma \ref{lem:noParams}, there exists a pair $(\HhH,\PPP)$ that $(\frac{\alpha}{6},\frac{\beta}{4})$-represent class $\CCC$, where 
$\size(\HhH) = O\Big(    \ln(1/\alpha) \cdot \big(  \RepDim(\CCC) + \ln\ln\ln(1/\alpha) + \ln\ln(1/\beta)  \big)  \Big)$.
Therefore, by Lemma \ref{lem:equivalence2}, there exists an algorithm $A$ that $(\alpha,\beta,\epsilon)$-PPAC learns the class $\CCC$ with a sample size 
$$m=O_{\beta}\left( \frac{1}{\alpha\epsilon}\ln(\frac{1}{\alpha})\cdot\left(  \RepDim(\CCC)+\ln\ln\ln(\frac{1}{\alpha})  \right) \right).$$

For the lower bound, let $A$ be an $(\alpha,\beta,\epsilon)$-PPAC learner for the class $\CCC$ with a sample size $m$, where $\alpha\leq 1/4$ and $\beta\leq 1/2$. By Lemma \ref{lem:equivalence1}, there exists an $(\HhH,\PPP)$ that $(\frac{1}{4},\frac{1}{4})$- probabilistically represents the class $\CCC$ and $\size(\HhH)=\ln(8)+\ln\ln(4)+ 8\alpha\epsilon m$.
Therefore, by definition, $\RepDim(\CCC)\leq\ln(8\ln(4))+ 8\alpha\epsilon m$. Thus,
$$m \geq \frac{1}{8\alpha\epsilon} \cdot \big( \RepDim(\CCC) - \ln(8\ln(4)) \big)=\Omega\left( \frac{\RepDim(\CCC)}{\alpha\epsilon} \right).$$
\end{proof}

\section{From a Probabilistic Representation to a Deterministic Representation}\label{sec:relationships}
In this section we will establish a connection between the (probabilistic) representation dimension of a class and its deterministic representation dimension.

\begin{observation}\label{obs:uniteH}
Let $(\HhH,\PPP)$ be an $(\alpha,\beta)$-probabilistic representation for a concept class $\CCC$. Then, $\BBB=\bigcup_{\HHH_i\in\HhH}\HHH_i$ is an $\alpha$-representation of $\CCC$.
\end{observation}

\begin{proof}
As $(\HhH,\PPP)$ is an $(\alpha,\beta)$-probabilistic representation for $\CCC$, for every $c$ and every $\DDD$
$$\Pr_{\PPP}[\exists h \in \HHH_i \;\; s.t \;\; \error_\DDD(c,h)\leq \alpha]\geq 1-\beta>0.$$
The probability is over choosing a set $\HHH_i\in_{\PPP}\HhH$. In particular, for every $c$ and every $\DDD$ there exists an $\HHH_i\in\HhH$ that contains an $\alpha$-good hypothesis.
\end{proof}

The simple construction in Observation~\ref{obs:uniteH} may result in a very large deterministic representation. For example, in Claim \ref{claim:pointEff} we show an $(\HhH,\PPP)$ that $(\alpha,\beta)$- probabilistically represents the class $\point_d$, where $\HhH$ contains all the sets of at most $\frac{4}{\alpha}\ln(\frac{1}{\beta})$ boolean functions. While $\bigcup_{\HHH_i\in\HhH}\HHH_i = 2^{X_d}$ is indeed an $\alpha$-representation for $\point_d$, it is extremely over-sized.

We will show that it is not necessary to take the union of all the $\HHH_i$'s in $\HhH$ in order to get an $\alpha$-representation for $\CCC$. As $(\HhH,\PPP)$ is an $(\alpha,\beta)$-probabilistic representation, for every $c$ and every $\DDD$, with probability at least $1-\beta$ a randomly chosen $\HHH_i \in_{\PPP} \HhH$ contains an $\alpha$-good hypothesis. The straight forward strategy here is to first boost $\beta$ as in Claim \ref{claim:boostB}, and then use the union bound over all possible $c\in \CCC$ and over all possible distributions $\DDD$ on $X_d$. Unfortunately, there are infinitely many such distributions, and the proof will be somewhat more complicated.

\begin{definition}\label{def:commit}
Let $\HhH=\{\HHH_1,\HHH_2,\ldots,\HHH_r\}$ be a family of hypothesis classes, and $\PPP$ be a distribution over $\{1,\ldots,r\}$. We will denote the following non private algorithm as  $Learner(\HhH,\PPP,m,\gamma)$:
$$\boxed{
\begin{array}{l}
\text{Input: a sample } S=(x_i,y_i)_{i=1}^m.\\
{\begin{array}{ll}
1. & \text{Randomly choose } \HHH_i \in_\PPP \HhH .\\
2. & \text{If for every } h \in \HHH_i \; \error_S(h) > \gamma \text{, then fail.} \\
3. & \text{Return } h \in \HHH_i \text{ minimizing } \error_S(h). \\\end{array}}\\
\end{array}}$$
We will say that $Learner(\HhH,\PPP,m,\gamma)$ is {\em $\beta$-successful} for a class $\CCC$ over $X_d$, if for every $c\in\CCC$ and every distribution $\DDD$ on $X_d$, given an input sample drawn i.i.d. according to $\DDD$ and labeled by $c$, algorithm $Learner$ fails with probability at most $\beta$.
\end{definition}

\begin{claim}\label{claim:shrinkClaimA}
If $(\HhH,\PPP)$ is an $(\alpha,\beta)$-probabilistic representation for a class $\CCC$, then, for $m\geq\frac{3}{\alpha}\ln(1/\beta)$, algorithm $Learner(\HhH,\PPP,m,2\alpha)$ is $2\beta$-successful for $\CCC$.
\end{claim}

\begin{proof}
We will show that with probability at least $1-2\beta$, the set $\HHH_i$ (chosen in Step~1) contains at least one hypothesis $h$ s.t.\ $\error_S(h)\leq2\alpha$. As $(\HhH,\PPP)$ is an $(\alpha,\beta)$-probabilistic representation for class $\CCC$, the chosen $\HHH_i$ will contain a hypothesis $h$ s.t. $\error_{\DDD}(c,h)\leq\alpha$ with probability at least $1-\beta$; by the Chernoff bound with probability at least $1-\exp(-m\alpha/3)$ this hypothesis has empirical error at most $2\alpha$. 
The set $\HHH_i$ contains a hypothesis $h$ s.t. $\error_S(h)\leq2\alpha$ with probability at least $(1-\beta)(1-\exp(-m\alpha/3))>1-(\beta+\exp(-m\alpha/3))$, which is at least $(1-2\beta)$ for $m\geq\frac{3}{\alpha}\ln(1/\beta)$.
\end{proof}

\begin{claim}\label{claim:shrinkClaimB}
Let $\HhH$ be a family of hypothesis classes, and $\PPP$ a distribution on it.
Let $\gamma,\beta$ and $m$ be such that $m\geq\frac{4}{\gamma} (\size(\HhH)+\ln(\frac{1}{\beta}))$.
If $Learner(\HhH,\PPP,m,\gamma)$ is $\beta$-successful for a class $\CCC$ over $X_d$,
then there exists $\widehat{\HhH} \subseteq \HhH$ and a distribution $\widehat{\PPP}$ on it, s.t. $Learner(\widehat{\HhH},\widehat{\PPP},m,\gamma)$ is a $(2\gamma,3\beta)$-PAC learner for $\CCC$ and $\left|\widehat{\HhH}\right| = \frac{d \cdot m}{\beta^2}$.
\end{claim}

\begin{proof}
For every input $S=(x_i,y_i)_{i=1}^m$, denote by $p_S$ the probability of $Learner(\HhH,\PPP,m,\gamma)$ failing on step 2 (the probability is only over the choice of $\HHH_i \in_\PPP \HhH$ in the first step).
As $Learner(\HhH,\PPP,m,\gamma)$ is $\beta$-successful,
$$\Pr_{\PPP,\DDD}\big[Learner(\HhH,\PPP,m,\gamma) \text{ fails}\big]=\sum_S{\Pr_\DDD[S] \cdot p_S} \leq \beta.$$
Consider the following process, denoted by Proc, for randomly choosing a multiset $\widetilde{\HhH}$ of size $t$ ($t$ will be set later):
$$\boxed{
\begin{array}{l}
\text{For } i=1, \ldots, t:\\
\hspace{5 mm} \text{Randomly choose } \HHH_i\in_\PPP \HhH\\
\text{Return } \widetilde{\HhH}=(\HHH_1,\HHH_2,...,\HHH_t).\\
\end{array}}$$
Denote by $\UUU_t$ the uniform distribution on $\{1,2,\ldots,t\}$. As before, for every input $S=(x_i,y_i)_{i=1}^m$, denote by $\widetilde{p_S}$ the probability of $Learner(\widetilde{\HhH},\UUU_t,m,\gamma)$ failing on its second step (again, the probability is only over the choice of $\HHH_i \in_{\UUU_t} \widetilde{\HhH}$ in the first step). Using those notations:
$$\Pr_{\UUU_t,\DDD}\big[Learner(\widetilde{\HhH},\UUU_t,m,\gamma) \text{ fails}\big]=\sum_S{\Pr_\DDD[S] \cdot \widetilde{p_S}}.$$
Fix a sample $S$. As the choice of $\HHH_i \in_{\UUU_t}\widetilde{\HhH}$ is uniform,
$$\widetilde{p_S}=\frac{\left|  \left\{\HHH_i\in\widetilde{\HhH} \; : \; \forall h\in \HHH_i \; \error_S(h)>\gamma \right\} \right|}{\left|\widetilde{\HhH}\right|}.$$
Using the Hoeffding bound,
$$\Pr_{Proc}\bigg[ \left|\widetilde{p_S} - p_S\right| \geq \beta \bigg] \leq 2 e^{-2t\beta^2}.$$
The probability is over choosing the multiset $\widetilde{\HhH}$.
There are at most $2^{m(d+1)}$ samples of size $m$
(as every entry in the sample is an element of $X_d$, concatenated with a label bit).
Using the union bound over all possible samples $S$,
$$\Pr_{Proc}\bigg[ \exists S \; s.t. \; \left|\widetilde{p_S} - p_S\right| \geq \beta \bigg] \leq 2^{m(d+1)} \cdot 2 \cdot e^{-2t\beta^2}.$$
For $t\geq \frac{m \cdot d}{\beta^2}$ the above probability is strictly less than 1. This means that for $t = \frac{m \cdot d}{\beta^2}$ there exists a multiset $\widehat{\HhH}$ such that $\left|\widehat{p_S} - p_S\right| \leq \beta$ for every sample $S$.
We will show that for this $\widehat{\HhH}$, $Learner(\widehat{\HhH},\UUU_t,m,\gamma)$ is a $(2\gamma,3\beta)$-PAC learner.
Fix a target concept $c\in\CCC$ and a distribution $\DDD$ on $X_d$. Define the following two good events:
\begin{enumerate}[label=$E_{\arabic*}$]
\item $Learner(\widehat{\HhH},\UUU_t,m,\gamma)$ outputs a hypothesis $h$ such that $\error_S(h)\leq\gamma$.
\item For every $h\in\HHH_i$ s.t. $\error_S(h)\leq\gamma$, it holds that $\error_{\DDD}(c,h)\leq2\gamma$.
\end{enumerate}
Note that if those two events happen, $Learner(\widehat{\HhH},\UUU_t,m,\gamma)$ returns a $2\gamma$-good hypothesis for $c$ and $\DDD$. We will show that those two events happen with high probability. We start by bounding the failure probability of $Learner(\widehat{\HhH},\UUU_t,m,\gamma)$.
\begin{eqnarray*}
\lefteqn{\Pr_{\UUU_t,\DDD}\big[Learner(\widehat{\HhH},\UUU_t,m,\gamma) \text{ fails}\big]} \\
&= &\sum_S{\Pr_\DDD[S] \cdot \widehat{p_S}} \\
&\leq &  \sum_S{\Pr_\DDD[S] \cdot (p_S+\beta)} \\
&= & \Pr_{\PPP,\DDD}\big[Learner(\HhH,\PPP,m,\gamma) \text{ fails}\big]+\beta \leq 2\beta.
\end{eqnarray*}
When $Learner(\widehat{\HhH},\UUU_t,m,\gamma)$ does not fail, it returns a hypothesis $h$
with empirical error at most $\gamma$. Thus, $\Pr[E_1]\geq1-2\beta$.

Using the Chernoff bound, the probability that a hypothesis $h$ with $\error_{\DDD}(c,h)>2\gamma$ has empirical error $\leq\gamma$ is less than $\exp(-m\gamma/4)$. Using the union bound, the probability that there is such a hypothesis in $\HHH_i$ is at most $|\HHH_i|\cdot\exp(-m\gamma/4)$. Therefore, $\Pr[ E_2 ]\geq1-|\HHH_i|\cdot\exp(-m\gamma/4)$. For $m\geq\frac{4}{\gamma} \ln(\frac{|\HHH_i|}{\beta})$, this probability is at least $(1-\beta)$.

All in all, the probability of $Learner(\HhH,\PPP,m,\gamma)$ failing to output a $2\gamma$-good hypothesis is at most $3\beta$.
\end{proof}

\begin{theorem}\label{thm:shrink}
If there exists a pair $(\HhH,\PPP)$ that $(\alpha,\beta)$-probabilistically represents a class $\CCC$ over $X_d$ (where $|\HhH|$ might be very big), then there exists a pair $(\widehat{\HhH},\widehat{\PPP})$ that $(4\alpha,6\beta)$-probabilistically represents $\CCC$, where $\widehat{\HhH}\subseteq\HhH$ , and
$$\left|\widehat{\HhH}\right| = \frac{3d}{4\alpha\beta^2}\left( \size(\HhH) + \ln(\frac{1}{\beta}) \right).$$
\end{theorem}

\begin{proof}
Let $(\HhH,\PPP)$ be an $(\alpha,\beta)$-probabilistic representation for a class $\CCC$.
Set $m=\frac{3}{\alpha}( \size(\HhH) + \ln(\frac{1}{\beta}) )$. By Claim \ref{claim:shrinkClaimA}, 
$Learner(\HhH,\PPP,m,2\alpha)$ is $2\beta$-successful for class $\CCC$. By Claim \ref{claim:shrinkClaimB}, there exists an $\widehat{\HhH} \subseteq \HhH$ and a distribution $\widehat{\PPP}$ on it, such that $Learner(\widehat{\HhH},\widehat{\PPP},m,2\alpha)$ is a $(4\alpha,6\beta)$-PAC learner for $\CCC$ and $\left|\widehat{\HhH}\right| = \frac{d \cdot m}{4\beta^2} = \frac{3d}{4\alpha\beta^2}( \size(\HhH) + \ln(\frac{1}{\beta}) ) $.

Assume towards contradiction that $(\widehat{\HhH},\widehat{\PPP})$ does not $(4\alpha,6\beta)$-represent $\CCC$. So, there exist a concept $c\in \CCC$ and a distribution $\DDD$ s.t., with probability strictly greater than $6\beta$, a randomly chosen $\HHH_i \in_{\widehat{\PPP}} {\widehat{\HhH}}$ does not contain a $4\alpha$-good hypothesis for $c,\DDD$. Therefore, for those $c$ and $\DDD$, $Learner(\widehat{\HhH},\widehat{\PPP},m,2\alpha)$ will fail to return a $4\alpha$-good hypothesis with probability strictly greater than $6\beta$.
\end{proof}

\begin{theorem}\label{thm:ABtoA}
For every class $\CCC$ over $X_d$ there exists a $\frac{1}{4}$-representation $\BBB$ such that
$\size(\BBB)=O( \ln(d) + \RepDim(\CCC))$.
\end{theorem}

\begin{proof}
By Lemma \ref{lem:noParams}, there exists a pair $(\HhH,\PPP)$ that
$(\frac{1}{16},\frac{1}{12})$-probabilistically represents $\CCC$ such that $\size(\HhH)=O( \RepDim(\CCC) )$. Using
Theorem \ref{thm:shrink}, there exists a pair $(\widehat{\HhH},\widehat{\PPP})$ that $(\frac{1}{4},\frac{1}{2})$-probabilistically represents $\CCC$, such that
$\size(\widehat{\HhH})=\size(\HhH)$ and
$$\left|\widehat{\HhH}\right| = O\left(d \cdot \size(\HhH) \right).$$

We can now use Observation~\ref{obs:uniteH} and construct the set
$\BBB=\bigcup_{\HHH_i\in\ \widehat{\HhH}}\HHH_i$ which
is a $\frac{1}{4}$-representation for the class $\CCC$. In addition,
$$ |\BBB|= O\left(    \left|\widehat{\HhH}\right|  \cdot e^{\size(\HhH)}   \right) =
O\left(    d \cdot \size(\HhH)  \cdot e^{\size(\HhH)}   \right).$$
\remove{$$=O\left(  \frac{d}{\alpha^2} \cdot \ln(\frac{1}{\alpha}) \left( \RepDim(\CCC) + \ln\ln\ln(\frac{1}{\alpha}) \right) \cdot \left(  \frac{1}{\alpha} \right)^{\RepDim(\CCC)+\ln\ln\ln(1/\alpha)}    \right).$$}
Thus,
$\size(\BBB)=\ln|\BBB|=O\left( \ln(d) + \RepDim(\CCC) \right)$.
\end{proof}

\begin{corollary}
For every concept class $\CCC$ over $X_d$, $\DRepDim(\CCC)=O(\ln(d)+\RepDim(\CCC))$.
\end{corollary}

\begin{corollary}
There exists a constant N s.t. for every concept class C over $X_d$ where $\DRepDim(\CCC) \geq N \log(d)$, the sample complexity that is necessary and sufficient for privately learning $\CCC$ is
 $\Theta_{\alpha,\beta}(\DRepDim(\CCC))$.
\end{corollary}

\section{Probabilistic \; Representation \; for Privately Solving Optimization Problems}

The notion of probabilistic representation applies not only to private
learning, but also to a broader task of optimization problems. We consider
the following scenario:

\begin{definition}
An {\em optimization problem} $\opt$ over a universe $X$ and a set of solutions $\FFF$ is defined by a quality function $q:X^*\times \FFF \rightarrow [0,1]$. Given a database $\db$, the task is to choose a solution $f\in\FFF$
such that $q(\db,f)$ is maximized.
\end{definition}

\myparagraph{Notation.} We will refer to the optimization problem defined by a quality function $q$ as $\opt_q$.

\begin{definition}
An {\em $\alpha$-good} solution for a database $\db$ is a solution $s$ such that $q(\db,s)\geq \max_{f\in\FFF}\{q(\db,f)\}-\alpha$. 
\end{definition}

Given an optimization problem $\opt_q$, one can use the exponential mechanism to choose a solution $s\in\FFF$.
In general, this method achieves a reasonable solution only for databases of size $\Omega(\log|\FFF|/\epsilon)$.
To see this, consider a case where there exists a database $\db$ of $m$ records such that exactly one solution $t\in\FFF$ has a quality of $q(\db,t)=1$, and every other $f\in\FFF$ has a quality of $q(\db,f)=1/2$. The probability of the exponential mechanism choosing $t$ is:
$$ \Pr[t \text{ is chosen}] = \frac{\exp(\epsilon m /2)}{(|\FFF|-1)\cdot\exp(\epsilon m /4)+\exp(\epsilon m /2)}.$$
Unless
\begin{eqnarray}
\label{eqn:neededGap}
  & m\geq\frac{4}{\epsilon}\ln(|\FFF|-1)=\Omega(\frac{1}{\epsilon}\ln|\FFF|),  &
\end{eqnarray}
the above probability is strictly less than $1/2$. 
Using our notations of probabilistic representation, it might be possible to reduce the necessary database size.

Consider using the exponential mechanism for choosing a solution $s$, not out of $\FFF$, but rather from a smaller set of solutions $\BBB$. Roughly speaking, the factor of $\ln|\FFF|$ in requirement (\ref{eqn:neededGap}) will now be replaced with $\ln|\BBB|$, which corresponds to size of the  representation.
Therefore, the database size $m$ should be at least $\ln|\BBB|/\epsilon$. So $m$ needs to be bigger than the size of the representation by at least a factor of $1/\epsilon$.

In the following analysis we will denote this required gap, i.e., $m/\ln|\BBB|$, as $\Delta$. We will see that the existence of a private approximation algorithm implies a probabilistic representation with $1<\Delta\approx\frac{1}{\epsilon}$, and that a probabilistic representation with $\Delta>1$ implies a private approximation algorithm. Bigger $\Delta$ corresponds to better privacy; however, it might be harder to achieve.

\begin{definition}\label{def:drepOpt}
Let $\opt_q$ be an optimization problem over a universe $X$ and a set of solutions $\FFF$.
Let $\BBB$ be a set of solutions, and denote $\size(\BBB)=\ln|\BBB|$.
We say that $\BBB$ is an $\alpha$-deterministic representation of $\opt_q$ 
for databases of $m$ elements if for every $\db\in X^m$ there exists a solution $s\in\BBB$ such that $q(\db,s)\geq \max_{f\in\FFF}\{q(\db,f)\}-\alpha$.
\end{definition}

\begin{definition}
Let $\BBB$ be an $\alpha$-deterministic representation of $\opt_q$ for databases of $m$ elements.
Denote $\Delta\triangleq\frac{m}{\size(\BBB)}$. If $\Delta>1$, then we say that the {\em ratio} of $\BBB$ is $\Delta$.
\end{definition}

An $\alpha$-deterministic representation $\BBB$ with ratio $\Delta$ is required to support all the databases of $m=\Delta\cdot\size(\BBB)$ elements. That is, for every $\db\in X^m$, the set $\BBB$ is required to contain at least one $\alpha$-good solution.

Fix $\db\in X^m$. Intuitively, $\Delta$ controls the ratio between $m$ and number of bits needed to represent an $\alpha$-good solution for $\db$. As $\BBB$ contains an $\alpha$-good solution for $\db$, and assuming $\BBB$ is publicly known, this solution could be represented with $\ln|\BBB|=\size(\BBB)=m/\Delta$ bits.

\begin{definition}\label{def:repOpt}
Let $\opt_q$ be an optimization problem over a universe $X$ and a set of solutions $\FFF$.
Let $\PPP$ be a distribution over $\{1,2, \ldots ,r\}$, and let $\BbB=\{\BBB_1,\BBB_2, \ldots ,\BBB_r\}$ be a family of solution sets for $\opt_q$.
We denote $\size(\BbB)=\max\{  \;  \ln|\BBB_i| : \BBB_i \in \BbB  \;  \}$.
We say that $(\BbB,\PPP)$ is an $(\alpha,\beta)$-probabilistic representation of $\opt_q$ 
for databases of $m$ elements if for every $\db\in X^m$:
$$\Pr_{\PPP}\left[\exists s \in \BBB_i \;\; s.t. \;\; q(\db,s)\geq \max_{f\in\FFF}\{q(\db,f)\}-\alpha\right]\geq 1-\beta.$$
\end{definition}

\begin{definition}
Let $(\BbB,\PPP)$ be an $(\alpha,\beta)$-probabilistic representation of $\opt_q$ for databases of $m$ elements.
Denote $\Delta\triangleq\frac{m}{\size(\BbB)}$. If $\Delta>1$, then we say that the {\em ratio} of the representation is $\Delta$.
\end{definition}

\begin{definition}
An optimization problem $\opt_q$ is {\em bounded} if $\Big||\db_1|\cdot q(\db_1,f)-|\db_2|\cdot q(\db_2,f)\Big| \leq 1$  for every solution $f$ and every two
neighboring databases $\db_1,\db_2$.
\end{definition}

We are interested in approximating bounded optimization problems, while guaranteeing differential privacy:

\begin{definition}
Let $\opt_q$ be a bounded optimization problem over a universe $X$ and a set of solutions $\FFF$.
An algorithm $A$ is an {\em $(\alpha,\beta,\epsilon)$-private approximation algorithm} for $\opt_q$ with a database of $m$ records if:
\begin{enumerate}
\item Algorithm $A$ is $\epsilon$-differentially private (as formulated in \defref{eps-dp});
\item For every $\db\in X^m$, algorithm $A$ outputs with probability at least $(1-\beta)$ a solution $s$ such that $q(\db,s)\geq \max_{f\in\FFF}\{q(\db,f)\}-\alpha$.
\end{enumerate}
\end{definition}

\begin{example}[Sanitization]
Consider a class of predicates $\CCC$ over $X$.
A database $\db$ contains points taken from $X$.
A predicate query $Q_c$ for $c\in\CCC$ is defined as $Q_c(\db) = \frac{1}{|\db|}\cdot |\{x_i \in \db \,:\,  c(x_i) =1\} |$.
Blum et al.~\cite{BLR08} defined a sanitizer (or data release mechanism) as a differentially private algorithm that, on input a database $\db$, outputs another database $\hat{\db}$ with entries taken from $X$.
A sanitizer $A$ is $(\alpha, \beta)$-useful for predicates in the class $\CCC$ if for every database $\db$ it holds that
$$ \Pr_A\left[\forall c\in C \;\; \big| Q_c(S)-Q_c(\hat{\db}) \big|\leq\alpha\right]\geq 1-\beta.$$

This scenario can be viewed as a bounded optimization problem:
The solutions are sanitized databases. For an input database $\db$ and and a sanitized database $\hat{\db}$, the quality function is
$$q(\db,\hat{\db})=1-\max_{c\in C}{\left\{|Q_c(\db)-Q_c(\hat{\db})|\right\}}.$$
To see that this optimization problem is bounded, note that for every two neighboring databases $\db_1,\db_2$ of $m$ elements, and every $c\in C$ it holds that $|Q_c(\db_1)-Q_c(\db_2)|\leq\frac{1}{m}$. Therefore, for every sanitized database $f$,
$$
m\cdot|q(\db_1,f)-q(\db_2,f)|
=m\cdot\left| \max_{c\in C}\{ |Q_c(\db_1)-Q_c(f)| \}-\max_{c\in C}\{ |Q_c(\db_2)-Q_c(f)| \}\right|
\leq 1
$$
\end{example}
\medskip

The next two lemmas establish an equivalence between 
a private approximation algorithm and a probabilistic 
representation for a bounded optimization problem.

\begin{lemma}\label{lem:optEq1}
Let $\opt_q$ be a bounded optimization problem over a universe $X$. If there exists a pair $(\BbB,\PPP)$ that $(\alpha,\beta)$-probabilistically represents $\opt_q$ for databases of $m$ elements, s.t. the ratio of $(\BbB,\PPP)$ is $\Delta>1$, then for every $\hat\alpha,\hat\beta,\epsilon$ satisfying
$$\Delta\geq\frac{2}{\epsilon\hat\alpha}\left(1+\frac{\ln(1/\hat\beta)}{\size(\BbB)}\right),$$
there exists an $\big( (\alpha+\hat\alpha),(\beta+\hat\beta),\epsilon \big)$-approximation algorithm for $\opt_q$ with a database of size 
$m$.
\end{lemma}

\begin{proof}
Consider the following algorithm $A$:
\begin{center}
\noindent\fbox{
\parbox{32pc}{
Inputs: a database $\db\in X^m$, and a privacy parameter $\epsilon$.

\begin{enumerate}[itemindent=3pt, topsep=1pt, itemsep=1pt,align=left,leftmargin=*,
        labelsep=5pt]
  \item Randomly choose $\BBB_i \in_\PPP \BbB$.
	\item Choose $s \in \BBB_i$ using the exponential mechanism, that is, with probability 
$$\frac{\exp(\epsilon\cdot m\cdot q(\db,s)/2)}{\sum_{f\in\BBB_i}\exp(\epsilon\cdot m\cdot q(\db,f)/2)}.$$
\end{enumerate}}}\end{center}

By the properties of the exponential mechanism, $A$ is $\epsilon$-differentially private. Fix a database $\db\in X^m$, and define the following 2 bad events:
\begin{enumerate}[label=$E_{\arabic*}$]
\item The set $\BBB_i$ chosen in step 1 does not contain a solution $s$ s.t. $q(\db,s)\geq \max_{f\in\FFF}\{q(\db,f)\}-\alpha$.
\item The solution $s$ chosen in step 2 is such that $q(\db,s) < \max_{t\in \BBB_i}{q(\db,t)} - \hat\alpha$.
\end{enumerate}
Note that if those two bad events do not occur, algorithm $A$ outputs a solution $s$ such that $q(\db,s)\geq\max_{f\in\FFF}\{q(\db,f)\}-\alpha-\hat\alpha$.
As $(\BbB,\PPP)$ is an $(\alpha,\beta)$-probabilistic representation of $\opt_q$ for databases of size $m$, event $E_1$ happens with probability at most $\beta$.
By the properties of the exponential mechanism, the probability of event $E_2$ is bounded by $|\BBB_i|\cdot\exp(-\epsilon m \hat\alpha/2)$.
As $m=\Delta\size(\BbB)$, this probability is at most
\begin{align*}
\Pr[E_2] &\leq \size(\BbB)\cdot\exp(-\epsilon m \hat\alpha/2)\\
&= \size(\BbB)\cdot\exp(-\epsilon \Delta\size(\BbB) \hat\alpha/2)\\
&\leq \size(\BbB)\cdot\exp\left(- \left(1+\frac{\ln(1/\hat\beta)}{\size(\BbB)}\right)\size(\BbB)\right)\\
&= \size(\BbB)\cdot\exp(-\size(\BbB)-\ln(1/\hat\beta)) = \hat\beta.
\end{align*}

Therefore, algorithm $A$ outputs an $(\alpha+\hat\alpha)$-good solution with probability at least $(1-\beta-\hat\beta)$.
\end{proof}

\begin{lemma}\label{lem:optEq2}
Let $\opt_q$ be an optimization problem. If there exists an 
$(\alpha,\beta,\epsilon)$-private approximation algorithm for $\opt_q$ with a database of $m$ records, then for every $\hat{\beta}$ satisfying
$$\Delta\triangleq\frac{m}{\ln(\frac{1}{1-\beta})+\ln\ln(\frac{1}{\hat\beta})+m\cdot\epsilon}>1,$$
there exists a pair $(\BbB,\PPP)$ that $(\alpha,\hat{\beta})$-probabilistically represents $\opt_q$ for databases of $m$ elements, where the ratio of the representation is $\Delta$.
\end{lemma}

\begin{proof}
Let $A$ be an $(\alpha,\beta,\epsilon)$-private approximation algorithm for $\opt_q$, with a sample size $m$.
Fix an arbitrary input database $\db\in X^m$. Define $G$ as the set of all solutions $s$, possibly outputted by $A$, such that $q(\db,s)\geq\max_{f\in\FFF}\{q(\db,f)\}-\alpha$.
As $A$ is an $(\alpha,\beta,\epsilon)$-approximation algorithm, 
$\Pr_A\left[ A(\db) \in G \right]\geq 1-\beta$.
As $A$ is $\epsilon$-differentially private, 
$\Pr_A\left[ A(\vec{0}) \in G \right] \geq (1-\beta) e^{-m\epsilon}$,
where $\vec{0}$ is a database with $m$ zeros. That is, $\Pr_A\left[ A(\vec{0}) \notin G \right] \leq 1-(1-\beta) e^{-m\epsilon}$. Now, consider a set $\BBB$ containing the outcomes of $ \Gamma \triangleq \frac{1}{1-\beta} \ln(\frac{1}{\hat{\beta}}) e^{m\epsilon} $ executions of $A(\vec0)$. The probability that $\BBB$ does not contain
a solutions $s\in G$
is at most $(1-(1-\beta) e^{-m\epsilon})^{ \Gamma  } \leq \hat{\beta}$. Thus,
$\BbB = \left\{ \BBB \subseteq support(A) \; : \; |\BBB| \leq  \Gamma  \right\}$,
and $\PPP$, the distribution induced by $A(\vec0)$, are an $(\alpha,\hat{\beta})$-probabilistic representation of $\opt_q$ for databases with $m$ elements. Moreover, the ratio of the representation is
\begin{eqnarray*}
\frac{m}{\size(\BbB)} &=& \frac{m}{\max\{  \;  \ln|\BBB| : \BBB \in \BbB  \;  \}}\\
&=&\frac{m}{\ln(\frac{1}{1-\beta})+\ln\ln(\frac{1}{\hat\beta})+m\epsilon}=\Delta.
\end{eqnarray*}
\end{proof}

\subsection{Exact 3SAT}
Consider the following bounded optimization problem, denoted as $\opt_{\rm E3SAT}$:
The universe $X$ is the set of all possible clauses with exactly $3$ different literals over $n$ variables, and the set of solutions $\FFF$ is the set of all possible $2^n$ assignments. Given a database $\db=(\sigma_1,\sigma_2,\ldots,\sigma_m)$ containing $m$ E3CNF clauses, the quality of an assignment $a\in\FFF$ is
$$q(\db,a)=\frac{|\{  i : a(\sigma_i)=1 \}|}{m}.$$

Aiming at the (very different) objective of secure protocols for search problems, Beimel et al.~\cite{BCNW08} defined the notation of solution-list algorithms, which corresponds to our notation of deterministic representation. We next rephrase their results using our notations.

\begin{enumerate}[label=$R{\arabic*}$]
\item For every $\alpha>0$ and every $\Delta>1$, there exists a set $\BBB$ that $(\alpha+1/8)$-deterministically represents $\opt_{\rm E3SAT}$ for databases of size $m=O\big(\Delta(\ln\ln(n)+\ln(1/\alpha)\big))$, and a ratio of $\Delta$.
\item Let $\alpha<1/2$ and $\Delta>1$. For every set $\BBB$ that $\alpha$- deterministically represents $\opt_{\rm E3SAT}$ for databases of size $m$ with a ratio of $\Delta$, it holds that $m=\Omega\big( \ln\ln(n)  \big)$.
\end{enumerate}

Using $(R1)$ and a deterministic version of Lemma \ref{lem:optEq1}, for every $\alpha,\beta,\epsilon>0$, there exists an 
$\big( (1/8+\alpha),\beta,\epsilon \big)$- approximation algorithm for $\opt_{\rm E3SAT}$ with a database of $m=O_{\alpha,\beta,\epsilon}(\ln\ln(n))$ clauses. By~$(R2)$, this is the best possible using a deterministic representation.

We can reduce the necessary database size, using a probabilistic representation. Fix a clause with three different literals. If we pick an assignment at random, then with
probability at least $7/8$ it satisfies the clause. Now, fix any exact 3CNF formula. If we pick an assignment
at random, then the expected fraction of satisfied clauses is at least $7/8$. Moreover, for every $0<\alpha<7/8$, the fraction of satisfied clauses is at least $(7/8-\alpha)$ with probability at least $\frac{\alpha}{\alpha+1/8}$. So, if we pick $t=\frac{\ln(1/\beta)}{\ln(\alpha+1/8)+\ln(1/\alpha)}$ random assignments, the probability that none of them will satisfy at least $(7/8-\alpha)m$ clauses is at most $\left(  \frac{\alpha}{\alpha+1/8}  \right)^t=\beta$. So, for every $\Delta>1$,
$$\BbB = \{ \BBB :  \BBB \text{ is a set of at most } t \text{ assignments}    \},$$
and $\PPP$, the distribution induced on $\BbB$ by randomly picking $t$ assignments, are an $\big(  (1/8+\alpha),\beta \big)$-probabilistic representation of $\opt_{\rm E3SAT}$ for databases of size $\Delta\cdot\ln(t)$ and a ratio of $\Delta$. By Lemma \ref{lem:optEq2}, for every $\epsilon$ there exists an 
$\big( (1/8+\alpha),\beta,\epsilon \big)$-approximation algorithm for $\opt_{\rm E3SAT}$ with a database of $m=o_{\alpha,\beta,\epsilon}(1)$ clauses.

\section{Extensions}
\subsection{$(\epsilon,\delta)$-Differential Privacy}
The notation of $\epsilon$-differential privacy was generalized to $(\epsilon,\delta)$-differential privacy, where the requirement in inequality (\ref{eqn:diffPrivDef}) is changed to
$$\Pr[A(\db_1 ) \in \mathcal{F}] \leq \exp(\epsilon) \cdot \Pr[A(\db_2) \in \mathcal{F}]+\delta.$$
The proof of Lemma \ref{lem:equivalence1} remains valid even if algorithm $A$ is only $(\epsilon,\delta)$-differential private for
\begin{eqnarray}
\label{eqn:deltaCondition}
  & \delta\leq\frac{1}{8} e^{-8\alpha\epsilon m} (1-e^{-\epsilon}).  &
\end{eqnarray}
To see this, note that inequality (\ref{eqn:epsilonDelta}) changes to

\begin{eqnarray*}
&&\Pr_A\left[ A(\vec{0}) \in G \right]\geq\\
&&\geq \left( \left( \left( \Pr_A\left[ A(\db) \in G \right] \cdot e^{-\epsilon} - \delta \right)e^{-\epsilon} - \delta \right) \cdots \right)e^{-\epsilon} - \delta \\
&&\geq \frac{1}{4} e^{-8\alpha\epsilon m}-\delta\left(\sum_{i=0}^{8\alpha m-1}{e^{-i\epsilon}}\right)\\
&&\geq \frac{1}{4} e^{-8\alpha\epsilon m}-\delta\left(\frac{1}{1-e^{-\epsilon}}\right) \geq \frac{1}{8} e^{-8\alpha\epsilon m}.
\end{eqnarray*}
The rest of the proof remains almost intact (only minor changes in the constants).
With that in mind, we see that the lower bound showed in Theorem \ref{thm:RepDim} 
for $\epsilon$-differentially private (that is, with $\delta=0$) learners
also applies for $(\epsilon,\delta)$-differentially private learners satisfying inequality (\ref{eqn:deltaCondition}). That is, every such learner for a class $\CCC$ must use $\Omega\left( \frac{\RepDim(\CCC)}{\alpha\epsilon}  \right)$ samples.

When using $(\epsilon,\delta)$-differential privacy, $\delta$ should be negligible in the security parameter, that is, in $d$ -- the representation length of elements in $X_d$. Therefore, using $(\epsilon,\delta)$-differential privacy instead of $\epsilon$-differential privacy cannot reduce the sample complexity for PPAC learning a concept class $\CCC$ whenever $\RepDim(\CCC)=O\left(  \log(d) \right)$.

\subsection{Probabilistic Representation Using a Hypothesis Class}
We will now consider a generalization of our representation notations that can be useful when considering PPAC learners that use a specific hypothesis class. In particular, those notation can be useful when considering proper-PPAC learners, that is, a learner that learns a class $\CCC$ using a hypothesis class $\BBB\subseteq\CCC$.
\begin{definition}
We define the $\alpha$-Deterministic Representation Dimension of a concept class $\CCC$ using a hypothesis class $\BBB$ as
$$\DRepDim_{\alpha}(\CCC,\BBB) = \min\left\{ \size(\HHH)  :
\begin{array}{l}
\HHH\subseteq\BBB \text{ is an }\\ 
\alpha \text{-representation}\\
\text{for class } \CCC
\end{array}\right\}.$$
\end{definition}

Note that $\DRepDim_{\frac{1}{4}}(\CCC,2^{X_d})=\DRepDim(\CCC)$. The dependency on $\alpha$ in the above definition is necessary: if $\CCC$ is not contained in $\BBB$ then for every small enough $\alpha$, the hypothesis class $\BBB$ itself does not $\alpha$-represents $\CCC$ (and therefore no subset $\HHH\subseteq\BBB$ can $\alpha$-represent $\CCC$). Moreover, when considering the notations of representation using a hypothesis class, our boosting technique for $\alpha$ does not work (as the boosting uses more complex hypotheses).

\begin{example}
Beimel et al.~\cite{BKN10} showed that for every $\alpha<1$, every subset $\HHH \subsetneq \point_d$ does not $\alpha$-represent the class $\point_d$. Therefore, $\DRepDim_{\alpha}(\point_d,\point_d) = \theta(d)$ for every $\alpha<1$.
\end{example}

\begin{definition}
A pair $(\HhH,\PPP)$ is an $(\alpha,\beta)$-probabilistic representation for a concept class $C$ {\em using a hypothesis class $\BBB$} if:
\begin{enumerate}
\item  $(\HhH,\PPP)$ is an  $(\alpha,\beta)$-probabilistic representation for the class $C$, as formulated in Definition \ref{def:prep}.
\item Every $\HHH_i\in\HhH$ is a subset of $\BBB$.
\end{enumerate}
\end{definition}

Note that whenever $\BBB=2^{X_d}$, this definition is identical to Definition \ref{def:prep}. Using this general notation, we can restate Lemma \ref{lem:equivalence2} and Lemma \ref{lem:equivalence1} as follows:
\begin{lemma}\label{lem:gen2}
If there exists a pair $(\HhH,\PPP)$ that $(\alpha,\beta)$- probabilistically represents a class $\CCC$ using a hypothesis class $\BBB$, then for every $\epsilon$ and every $\gamma$ there exists an algorithm $A$ that $(\alpha+\gamma,3\beta,\epsilon)$-PPAC learns $\CCC$ using $\BBB$ and a sample size
$m=O( ( \size(\HhH) + \ln(\frac{1}{\beta}) ) \max\{ \frac{1}{\gamma \epsilon} , \frac{1}{ \gamma^2 }  \} )$.
\end{lemma}

Note that in the above lemma the resulting algorithm $A$ has accuracy $(\alpha+\gamma)$ as opposed to $6\alpha$ in lemma \ref{lem:equivalence2}, where $\gamma$ is arbitrary. While in section \ref{sec:learn} we did not mind the multiplicative factor of $6$ in the accuracy parameter (as we could boost it back), replacing it with an additive factor of $\gamma$ might be of value in this section as our
boosting technique for the accuracy parameter does not work here. As an example, consider a representation with $\alpha=\frac{1}{10}$. Without boosting capabilities, this change makes the difference between the ability to generate an algorithm with $\alpha=\frac{6}{10}$, or an algorithm with $\alpha=\frac{1}{10}+\frac{1}{1000}$.

\begin{proof}
Let $(\HhH,\PPP)$ be an $(\alpha,\beta)$-probabilistic representation for class $\CCC$ using a hypothesis class $\BBB$, and consider the following algorithm $A$:
$$\boxed{
\begin{array}{l}
\text{Inputs: } S=(x_i,y_i)_{i=1}^m \text{, and a privacy parameter } \epsilon.\\
{\begin{array}{ll}
1. & \text{Randomly choose } \HHH_i \in_\PPP \HhH .\\
2. & \text{Choose } h \in \HHH_i \text{ using the exp. mechanism with } \epsilon.\\
\end{array}}\\
\end{array}}$$
First note that the support of $A$ is indeed (a subset of) $\BBB$.
By the properties of the exponential mechanism, $A$ is $\epsilon$-differentially private. Fix some $c\in \CCC$ and $\DDD$, and define the following 3 good events:
\begin{enumerate}[label=$E_{\arabic*}$]
\item $\HHH_i$ chosen in step 1 contains at least one hypothesis $h$ s.t. $\error_{\DDD}(h)\leq\alpha$.
\item For every $h\in\HHH_i$ it holds that $|\error_S(h)-\error_{\DDD}(c,h)|\leq\frac{\gamma}{3}$.
\item The exponential mechanism chooses an $h$ such that $\error_S(h) \leq \frac{\gamma}{3} + \min_{f\in \HHH_i}\left\{\error_S(f)\right\}$.
\end{enumerate}
Note that if those 3 good events happen, algorithm $A$ returns an $(\alpha+\gamma)$-good hypothesis. We will now show that those 3 events happen with high probability.

As $(\HhH,\PPP)$ is an $(\alpha,\beta)$-probabilistic representation for the class $\CCC$, event $E_1$ happens with probability at least $1-\beta$.

Using the Hoeffding bound, event $E_2$ happens with probability at leat $1-2|\HHH_i|\exp(-\frac{2}{9}\gamma^2 m)$. For $m\geq\frac{9}{2\gamma^2}\ln(\frac{2|\HHH_i|}{\beta})$, this probability is at leat $1-\beta$. 

The exponential mechanism ensures that the probability of event $E_3$ is at least $1-|\HHH_i| \cdot \exp(-\epsilon \gamma m /6)$ (see Section \ref{expMech}), which is at least $(1-\beta)$ for $m \geq \frac{6}{\gamma \epsilon} \ln(\frac{|\HHH_i|}{\beta})$.

All in all, by setting $m=6( \size(\HhH) + \ln(\frac{2}{\beta}) )\max\{\frac{1}{\gamma^2},\frac{1}{\gamma\epsilon} \} $ we ensure that the probability of $A$ failing to output an $(\alpha+\gamma)$-good hypothesis is at most $3\beta$.
\end{proof}

\begin{lemma}\label{lem:gen1}
If there exists an algorithm $A$ that $(\alpha,\frac{1}{2},\epsilon)$-PPAC learns a concept class $\CCC$ using a hypothesis class $\BBB$ and a sample size $m$, then there exists a pair $(\HhH,\PPP)$ that $(\alpha,1/4)$-probabilistically represents the class $\CCC$ using the hypothesis class $\BBB$ where $\size(\HhH) = O\left(m\epsilon\right)$.
\end{lemma}

The proof of Lemma~\ref{lem:gen1} is identical to the proof of Lemma~\ref{lem:oldEquivalence1}.

\begin{definition}
We define the $\alpha$-Probabilistic Representation Dimension of a concept class $\CCC$ using a hypothesis class $\BBB$ as
$$\RepDim_{\alpha}(\CCC,\BBB) = \min\left\{ \size(\HhH)  : 
\begin{array}{l}
 \exists \PPP \text{ s.t. } (\HhH,\PPP)\\
 \text{ is an } (\alpha,\frac{1}{4})\text{-prob.}\\
 \text{representation}\\
 \text{for } \CCC \text{ using } \BBB
\end{array}
\right\}.$$
\end{definition}

\begin{example}
Beimel et al.~\cite{BKN10} showed that for every $\alpha<1$, every proper-PPAC learner for $\point_d$ requires $\Omega( (d+\log(1/\beta))/(\epsilon \alpha) )$ labled examples. Using Lemma \ref{lem:gen2}, we get that $\RepDim_{\alpha}(\point_d,\point_d)=\Omega(d)$.
\end{example}

We still do not know the relation between the representation dimension of a concept class and its VC dimension. However, the above example shows a strong separation between the VC dimension of the class $\point_d$ and  $\RepDim_{\alpha}(\point_d,\point_d)$.

\section{A Probabilistic Representation for Points}\label{sec:proofs}

Example \ref{example:pointRep} states the existence of a constant size probabilistic representation for the class $\point_d$. We now give the construction.
\begin{claim}\label{claim:pointEff}
There exists an $(\alpha,\beta)$-probabilistic representation for $\point_d$ of $\size$ $\ln(4/\alpha)+\ln\ln(1/\beta)$.
Furthermore, each hypothesis $h$ in each $\HHH_i$ has a short description and given $x$, the value $h(x)$ can be computed efficiently.
\end{claim}
\begin{proof}
Consider the following set of hypothesis classes
$$\HhH=\left\{ \HHH \subseteq 2^{X_d} \; : \; |\HHH|\leq \frac{4}{\alpha} \ln(\frac{1}{\beta}) \right\}.$$
That is, $\HHH\in\HhH$ if $\HHH$ contains at most $\frac{4}{\alpha} \ln(\frac{1}{\beta})$ boolean functions. We will show an appropriate distribution $\PPP$ s.t. $(\HhH,\PPP)$ is an $(\alpha,\beta)$-probabilistic representation of the class $\point_d$.
To this end, fix a target concept $c_j \in \point_d$ and a distribution $\DDD$ on $X_d$ (remember that $j$ is the unique point on which $c_j(j)=1$). We need to show how to randomly choose an $\HHH\in_R \HhH$ such that with probability at least $(1-\beta)$ over the choice of $\HHH$, there will be at least one $h\in \HHH$ such that $\error_\DDD(c_j,h) \leq \alpha$. Consider the following process for randomly choosing an $\HHH\in\HhH$:

$$\boxed{
\begin{array}{l}
\text{1. Denote } M=\frac{4}{\alpha} \ln(\frac{1}{\beta})\\
\text{2. For } i=1, \ldots, M \text{ construct hypothesis } h_i \text{ as follows:}\\
\hspace{10 mm} \text{For each } x\in X_d \text{ (independently):}\\
\hspace{20 mm} \text{Let } h_i(x)=1 \text{ with probability } \alpha/2,\\
\hspace{19 mm} \text{ and } h_i(x)=0 \text{  otherwise. }\\
\text{3. Return } \HHH=\{h_1,h_2,\ldots,h_M \}.
\end{array}}$$

The above process induces a distribution on $\HhH$, denoted as $\PPP$.
We will next analyze the probability that the returned $\HHH$ does not contain an $\alpha$-good hypothesis. We start by fixing some $i$ and analyzing the expected error of $h_i$, conditioned on the event that $h_i(j)=1$. The probability is taken over the random coins used to construct $h_i$.
\begin{eqnarray*}
&&\E_{h_i}\left[\error_\DDD(c_j,h_i) \; \Big| \; h_i(j)=1\right]=\\
&&=\E_{h_i}\left[ \E_{x\in\DDD}\left[ \, \big|c_j(x)-h_i(x)\big| \, \right] \;\Big|\; h_i(j)=1\right]\\
&&=\E_{x\in\DDD}\left[ \E_{h_i}\left[ \, \big|c_j(x)-h_i(x)\big| \;\Big|\; h_i(j)=1 \right] \right] \leq \frac{\alpha}{2}.
\end{eqnarray*}
Using Markov's Inequality,
$$\Pr_{h_i}\left[\error_\DDD(c_j,h_i) \geq \alpha \; \bigg| \; h_i(j)=1\right] \leq \frac{1}{2}.$$
So, the probability that $h_i$ is $\alpha$-good for $c_j$ and $\DDD$ is:
\begin{eqnarray*}
&&\Pr_{h_i}\left[\error_\DDD(c_j,h_i)\leq\alpha\right]\geq\\
&&\geq \Pr_{h_i}\left[h_i(j)=1\right] \cdot \Pr_{h_i}\left[\error_\DDD(c_j,h_i)\leq\alpha \;\bigg|\; h_i(j)=1\right]\\
&&\geq \frac{\alpha}{2} \cdot \frac{1}{2}=\frac{\alpha}{4}.
\end{eqnarray*}
Thus, the probability that $\HHH$ fails to contain an $\alpha$-good hypothesis is at most
$\left(1-\frac{\alpha}{4}\right)^M$, which is less than $\beta$ for our choice of $M$. This concludes the proof that $(\HhH,\PPP)$ is an $(\alpha,\beta)$-probabilistic representation for $\point_d$.

When a hypothesis $h_i()$ was constructed in the above random process, the value of $h_i(x)$ was independently drawn for every $x\in X_d$. This results in a hypothesis whose description size is $O(2^d)$, which in turn, will result in a non efficient learning algorithm.  We next construct hypotheses whose description is short. To achieve this goal, we note that in the above analysis we only care about the probability that $h_i(x)=0$ given that $h_i(j)=1$. Thus, we can choose the values of $h_i$ in a pairwise independent way, e.g., using a random polynomial of degree 2. The size of the description in this case is $O(d)$.
\remove{We can represent (and efficiently sample) our hypotheses using only two elements of the finite field $\F_p$ for some prime $p>\max\{\frac{2}{\alpha},2^d \}$. That is, we can replace the distribution $\PPP$ with the distribution $\widehat{\PPP}$ induced by the following random process:
$$\boxed{
\begin{array}{ll}
1.&\text{Denote } M=\frac{4}{\alpha} \ln(\frac{1}{\beta})\\
2.&\text{Denote } p>\max\{\frac{2}{\alpha},2^d \} \text{ an arbitrary prime.}\\
3.& \text{For } i=1, \ldots, M \text{ build a hypothesis } h_i \text{ as follows:}\\
&\hspace{5 mm} \text{Randomly draw } a_i,b_i\in_R\F_p   \text{ (independently).}\\
&\hspace{5 mm} h_i(x)=\begin{cases}
1, & a_i x + b_i < \frac{\alpha\cdot p}{2} \text{ (mod p)}\\
0, & \text{otherwise}
\end{cases} \\
4.&\text{Return } \HHH=\{h_1,h_2,\ldots,h_M \}.
\end{array}}$$  }
\end{proof}

\begin{observation}
Consider the class $\point_{\N}$, defined in Example \ref{example:pointRatio}. The above construction can be adjusted to yield an (inefficient) improper private learner for $\point_{\N}$ with $O_{\alpha,\beta,\epsilon}(1)$ samples.
The only adjustments necessary are in the construction of the $(\alpha,\beta)$-probabilistic representation. Specifically, we need to specify how to randomly draw a boolean function $h$ over the natural numbers, such that for every $x\in\N$ the probability of $h(x)=1$ is $\alpha/2$, and the values of $h$ on every two distinct points in $\N$ are independent. This can be done easily, as a random real number could be interpreted as a random function over $\N$.
\end{observation}

\bibliographystyle{abbrv}

\begin{thebibliography}{10}

\bibitem{BBKN12}
A.~Beimel, H.~Brenner, S.~P. Kasiviswanathan, and K.~Nissim.
\newblock Bounds on the sample complexity for private learning and private data
  release.
\newblock Machine learning, 2013. Full version of~\cite{BKN10}.

\bibitem{BCNW08}
A.~Beimel, P.~Carmi, K.~Nissim, and E.~Weinreb.
\newblock Private approximation of search problems.
\newblock {\em SIAM J. Comput.}, 38(5):1728--1760, 2008.

\bibitem{BKN10}
A.~Beimel, S.~P. Kasiviswanathan, and K.~Nissim.
\newblock Bounds on the sample complexity for private learning and private data
  release.
\newblock In {\em TCC}, volume 5978 of {\em LNCS}, pages 437--454. Springer,
  2010.

\bibitem{BNS13}
A.~Beimel, K.~Nissim, and U.~Stemmer.
\newblock Characterizing the sample complexity of private learners.
\newblock In {\em ITCS}, pages 97--110, 2013.

\bibitem{BDMN05}
A.~Blum, C.~Dwork, F.~McSherry, and K.~Nissim.
\newblock Practical privacy: The {SuLQ} framework.
\newblock In {\em PODS}, pages 128--138. ACM, 2005.

\bibitem{BLR08}
A.~Blum, K.~Ligett, and A.~Roth.
\newblock A learning theory approach to non-interactive database privacy.
\newblock In {\em STOC}, pages 609--618. ACM, 2008.

\bibitem{BEHW}
A.~Blumer, A.~Ehrenfeucht, D.~Haussler, and M.~K. Warmuth.
\newblock {Learnability and the Vapnik-Chervonenkis dimension}.
\newblock {\em ACM}, 36(4):929--965, 1989.

\bibitem{CH11}
K.~Chaudhuri and D.~Hsu.
\newblock Sample complexity bounds for differentially private learning.
\newblock {\em COLT}, 19:155--186, 2011.

\bibitem{chern}
H.~Chernoff.
\newblock A measure of asymptotic efficiency for tests of a hypothesis based on
  the sum of observations.
\newblock {\em Ann. Math. Statist.}, 23:493--507, 1952.

\bibitem{Dwork09}
C.~Dwork.
\newblock The differential privacy frontier.
\newblock In O.~Reingold, editor, {\em TCC}, volume 5444 of {\em LNCS}, pages
  496--502. Springer, 2009.

\bibitem{Dwork11}
C.~Dwork.
\newblock A firm foundation for private data analysis.
\newblock {\em Commun. of the ACM}, 54(1):86--95, 2011.

\bibitem{DMNS06}
C.~Dwork, F.~McSherry, K.~Nissim, and A.~Smith.
\newblock Calibrating noise to sensitivity in private data analysis.
\newblock In S.~Halevi and T.~Rabin, editors, {\em TCC}, volume 3876 of {\em
  LNCS}, pages 265--284. Springer, 2006.

\bibitem{DRV10}
C.~Dwork, G.~N. Rothblum, and S.~P. Vadhan.
\newblock Boosting and differential privacy.
\newblock In {\em FOCS}, pages 51--60, 2010.

\bibitem{EHKV}
A.~Ehrenfeucht, D.~Haussler, M.~J. Kearns, and L.~G. Valiant.
\newblock A general lower bound on the number of examples needed for learning.
\newblock {\em Inf. Comput.}, 82(3):247--261, 1989.

\bibitem{Adaboost}
Y.~Freund and R.~E. Schapire.
\newblock A decision-theoretic generalization of on-line learning and an
  application to boosting.
\newblock {\em Journal of Computer and System Sciences}, 55(1):119 -- 139,
  1997.

\bibitem{hoeff}
W.~Hoeffding.
\newblock Probability inequalities for sums of bounded random variables.
\newblock {\em Journal of the American Statistical Association},
  58(301):13--30, 1963.

\bibitem{KLNRS08}
S.~P. Kasiviswanathan, H.~K. Lee, K.~Nissim, S.~Raskhodnikova, and A.~Smith.
\newblock What can we learn privately?
\newblock {\em SIAM J. Comput.}, 40(3):793--826, 2011.

\bibitem{MT07}
F.~McSherry and K.~Talwar.
\newblock Mechanism design via differential privacy.
\newblock In {\em FOCS}, pages 94--103. IEEE, 2007.

\bibitem{Schapire90}
R.~E. Schapire.
\newblock The strength of weak learnability.
\newblock {\em Mach. Learn.}, 5(2):197--227, 1990.

\bibitem{Valiant84}
L.~G. Valiant.
\newblock A theory of the learnable.
\newblock {\em Communications of the ACM}, 27:1134--1142, 1984.

\bibitem{VC}
V.~N. Vapnik and A.~Y. Chervonenkis.
\newblock {On the uniform convergence of relative frequencies of events to
  their probabilities}.
\newblock {\em Theory of Probability and its Applications}, 16:264, 1971.

\end{thebibliography}

% sigproc.bib is the name of the Bibliography in this case
% ACM needs 'a single self-contained file'!
%
\end{document}